\documentclass[sigconf]{acmart}
\AtBeginDocument{%
  \providecommand\BibTeX{{%
    \normalfont B\kern-0.5em{\scshape i\kern-0.25em b}\kern-0.8em\TeX}}}

\usepackage{soul}
\usepackage{graphicx}
\usepackage{subfig}
\usepackage{comment}
\usepackage{caption}
\usepackage{color}
\usepackage{mathtools}
\usepackage{bm}
\usepackage{array}
\usepackage{lscape}
\usepackage{multirow}
\usepackage[ruled,vlined,linesnumbered]{algorithm2e}
\usepackage{url}
\usepackage{amsthm}
\usepackage{epstopdf}
\theoremstyle{definition}

\DeclareMathOperator{\EX}{\mathbb{E}} 
\DeclarePairedDelimiter\ceil{\lceil}{\rceil}
\DeclarePairedDelimiter\floor{\lfloor}{\rfloor}

\setcopyright{acmcopyright}
\copyrightyear{2021}
\acmYear{2021}
\setcopyright{acmcopyright}\acmConference[SSDBM 2021]{33rd International Conference on Scientific and Statistical Database Management}{July 6--7, 2021}{Tampa, FL, USA}
\acmBooktitle{33rd International Conference on Scientific and Statistical Database Management (SSDBM 2021), July 6--7, 2021, Tampa, FL, USA}
\acmPrice{15.00}
\acmDOI{10.1145/3468791.3468822}
\acmISBN{978-1-4503-8413-1/21/07}

\begin{document}

\title{Sub-trajectory Similarity Join with Obfuscation}


\author{Yanchuan Chang}
\affiliation{%
  \institution{The University of Melbourne}
  \country{Australia}
}
\email{yanchuanc@student.unimelb.edu.au}

\author{Jianzhong Qi}
\affiliation{%
  \institution{The University of Melbourne}
  \country{Australia}
}
\email{jianzhong.qi@unimelb.edu.au}

\author{Egemen Tanin}
\affiliation{%
  \institution{The University of Melbourne}
  \country{Australia}
 }
\email{etanin@unimelb.edu.au}

\author{Xingjun Ma}
\affiliation{%
  \institution{Deakin University}
  \country{Australia}
}
\email{daniel.ma@deakin.edu.au}

\author{Hanan Samet}
\affiliation{%
  \institution{University of Maryland}
  \country{USA}
 }
\email{hjs@cs.umd.edu}

\renewcommand{\shortauthors}{Yanchuan Chang et al.}

\begin{abstract}
User trajectory data is becoming increasingly accessible due to the prevalence of GPS-equipped devices such as smartphones. Many existing studies focus on querying trajectories that are similar to each other in their entirety. We observe that trajectories partially similar to each other contain useful information about users' travel patterns which should not be ignored. Such partially similar trajectories are critical in applications such as epidemic contact tracing. We thus propose to query trajectories that are within a given distance range from each other for a given period of time. We formulate this problem as a sub-trajectory similarity join query named as the \emph{STS-Join}. We further propose a distributed index structure and a query algorithm for STS-Join, where users retain their raw location data and only send obfuscated trajectories to a server for query processing. This helps preserve user location privacy which is vital when dealing with such data. Theoretical analysis and experiments on real data confirm the effectiveness and the efficiency of our proposed index structure and query algorithm.



\end{abstract}





\maketitle

\section{Introduction}

Trajectory data is being captured by GPS-equipped devices such as smartphones. Such data can be used to query people's travel patterns. In this paper, we are interested in a type of query named as the \emph{trajectory join} query, which returns all trajectory pairs from two trajectory sets that satisfy a given join predicate, e.g., finding people with similar commute routes for ride-sharing matches. 
Most existing trajectory join queries~\cite{join_dita, simi_shuoshang_join, join_dison} compute trajectories that are similar in their entirety, i.e., their join predicates are defined on the full trajectories. 

We observe that trajectories that are partially similar to each other also offer useful information and should not be ignored. 
Such partially similar trajectories are gaining importance in applications such as contact tracing for managing epidemics, e.g., to find people in close contact with confirmed cases of COVID-19 for a duration of over 15 minutes\footnote{www.dhhs.vic.gov.au/victorian-public-coronavirus-disease-covid-19}. Another example is to compute partially similar trajectories to find matches to form \emph{multi-hop} goods delivery or car-pooling arrangements that allow transits~\cite{multihop}.

Motivated by these applications, we propose a trajectory join query that, given two sets of trajectories, computes every pair of trajectories that are within a distance threshold  (e.g., 5 meters) lasting for a certain time span (e.g., 15 minutes). 
Our query is defined on \underline{\emph{s}}ub-\underline{\emph{t}}rajectory \underline{\emph{s}}imilarity and hence is 
named as the \emph{STS-Join}.

\begin{figure}[htp]
  \centering
  \includegraphics[width=0.37\textwidth]{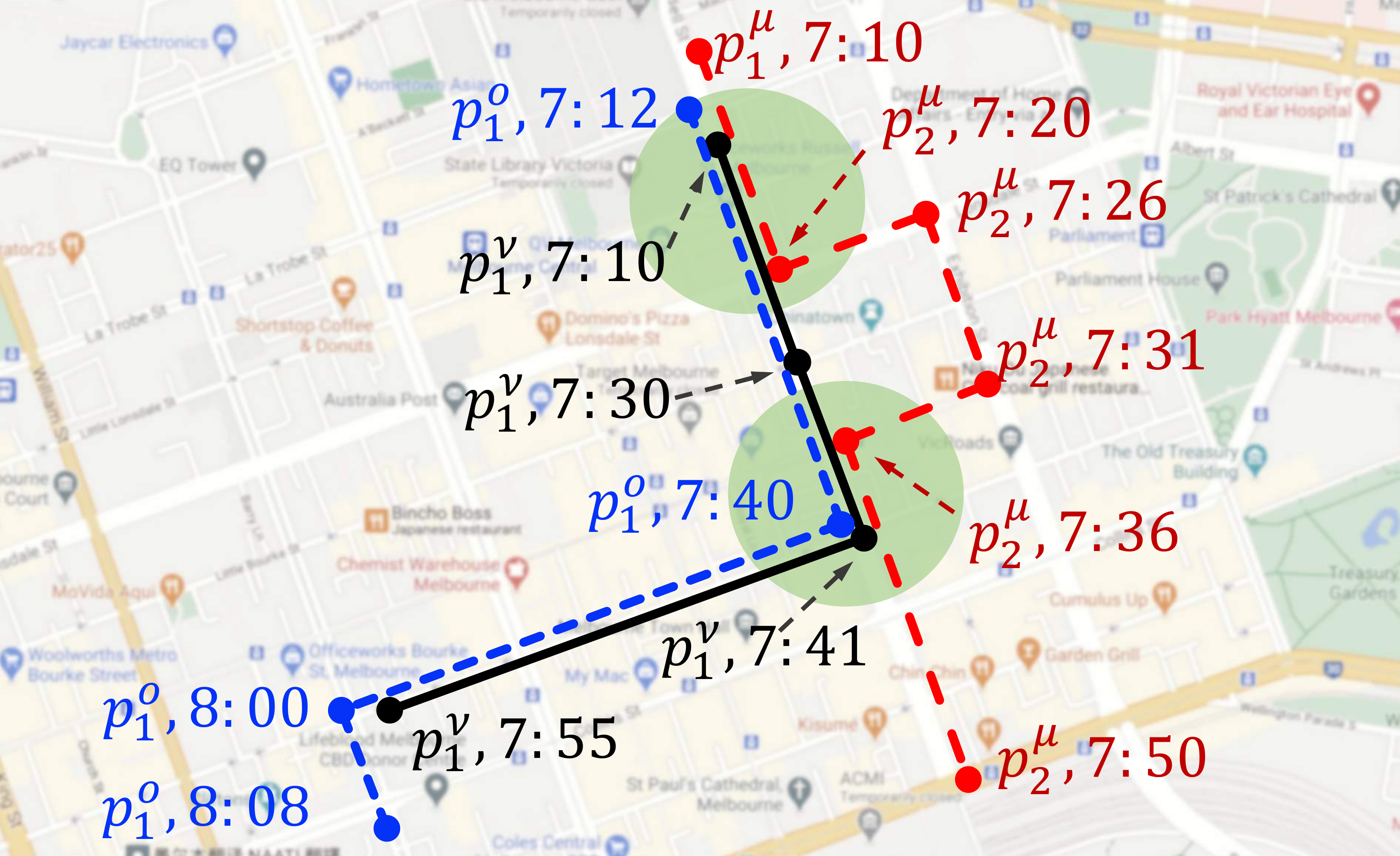}
  \caption{Trajectory join examples}\label{fig:intro_map_example}
\end{figure}

Fig.~\ref{fig:intro_map_example} illustrates our join predicate. There are three trajectories in different colors. The sampled points on the trajectories are marked by the dots and are labelled with their time. Existing studies on full trajectory similarity may return the black and the blue trajectories as a similar pair, because they are very close in both space and time. The existing studies may lose sight of similarity of portions. Within the two green colored areas, the red trajectory is close to the blue and the black ones, respectively, even though their entire movements are quite dissimilar. STS-Join aims to find out both partially and fully similar trajectories.

We assume a distributed (i.e., client-server) system  architecture to process STS-Join queries. 
Each client device (e.g., a user's mobile phone) stores a user's own accurate trajectories, while the server only stores a modified version of the trajectories for privacy considerations. To showcase the feasibility of processing STS-Join queries over modified trajectories, we consider  trajectory obfuscation. A user  trajectory is obfuscated automatically where every sampled point on the trajectory is shifted  with a bounded-distance noise before the trajectory is sent to the server. All users' obfuscated trajectories are then maintained on the server in a spatio-temporal index for fast query retrieval. We note that there are limitations in using only obfuscation for trajectory privacy protection, and more advanced techniques exist in the literature such as~\emph{geo-indistinguishability}~\cite{dp_geo_indistinguish}. However, the core theme and contribution of our study is \emph{not} to propose another privacy protection technique. Thus, we just use obfuscation for its simplicity and leave more advanced privacy protection techniques for future studies. 



We propose a query algorithm for STS-Join based on a traversal over our index structure. To take advantage of the characteristic that segments on a trajectory are connected sequentially, we design a backtracking-based method to reduce node access of the traversal. It can avoid querying each segment individually. Note that this query method is applicable to any spatial indices  that divide the space in a non-overlapping manner.
Further, we derive an upper bound of the similarity between  
a query trajectory and an original user trajectory based on the similarity between the query trajectory and the corresponding obfuscated user trajectory. This enables additional pruning on the server, which reduces the number of query trajectories to be sent to the clients for final similarity checks and saves communication costs.


To sum up, we make the following contributions:
\begin{enumerate}
    \item We define a new trajectory join predicate STS-Join that focuses on sub-trajectory similarity. Two similar trajectories can be very different, but they can contain parts that are related in both space and time.
    This similarity is especially applicable to 
    bioinformatic datasets that are used for contact tracing and in computational transport science with  shared-economy-based transportation systems. 
    \item We propose an efficient  spatio-temporal index to manage trajectory data dynamically and a backtracking-based algorithm to process STS-Join queries. We further propose  a similarity upper bound that is computed on obfuscated trajectories to enable data pruning and more efficient STS-Join processing. 
    Trajectories in our index are not required to be accurate, but our join results are still correct.
    \item We conduct experiments on real datasets. The proposed join algorithm outperforms adapted state-of-the-art methods by up to three orders of magnitude in running time.
\end{enumerate}
\section{Related Work} \label{sec:relatedwork}

Our study is related to studies on trajectory similarity measurements, trajectory join queries, and trajectory privacy.

\subsection{Trajectory Similarity Measurement}
Most trajectory similarity measurements are either spatial distance based or spatio-temporal distance based.

Spatial-based measurements~\cite{join_dita, simi_lcss_example, simi_edr, simi_ftw, simi_fastdtw} focus on the spatial distance between two trajectories.
They aggregate distance between aligned point pairs form two trajectories, such as \emph{dynamic time warping} (DTW)\cite{join_dita} , \emph{longest common subsequence}  (LCSS)\cite{simi_lcss_example} and \emph{edit distance on real sequence} (EDR)\cite{simi_edr}. 

Spatio-temporal-based measurements~\cite{simi_lcss, simi_stEDR, simi_STeuclidean, simi_shuoshang_join,DBLP:conf/icde/WuXCKNK16} consider the distance in both space and time. 
For example, 
LCSS and EDR have been extended to incorporate temporal thresholds~\cite{simi_lcss,simi_stEDR}. 
Nanni and Pedreschi~\cite{simi_STeuclidean} assume a constant moving speed on each segment of a trajectory, which is computed as the length of the segment divided by its time span. They then compute the distance between two trajectories as the average distance between two users who travel along the two trajectories with the constant speed of each segment. 
Shang et al.~\cite{simi_shuoshang_join} sum point-to-trajectory distances from one trajectory to the other as the distance between two trajectories, where the summed distance is a weighted sum of the spatial and temporal distances. Wu et al.~\cite{DBLP:conf/icde/WuXCKNK16} consider points from two trajectories to be compatible if their spatial distance over time difference is within a velocity threshold. 

These measurements compute similarity based on full trajectories which are different from our sub-trajectory-based metric. 

\subsection{Trajectory Join}\label{sec:relate_join}

Studies leverage distributed structures to join trajectories, such as 
\emph{DITA}~\cite{join_dita} and \emph{DISON}~\cite{join_dison}. 
DITA supports a variety of trajectory distance functions based on full trajectories, while DISON measures trajectory similarity by counting the length of common road segments among the trajectories. 
Our trajectory join procedure supports distributed processing in two senses: (i) our join refinement procedure is distributed to client machines, and (ii) our index is based on non-overlapping space/time partitioning, which can be easily distributed.

There are also studies on sub-trajectory join~\cite{join_subtraj_edbt2020, join_continuousdistributed,join_CSTJ}.  
\emph{CSTJ}~\cite{join_CSTJ} joins trajectories online. It returns two trajectories once they stay as close-distance pairs 
for a given time threshold. Its distance measure is based on points in trajectories rather than segments as in our study. 
\emph{DTJr}~\cite{join_continuousdistributed} also measures point distance. It finds the \emph{longest} sub-trajectory pair satisfying both spatial and temporal distance thresholds, while we find \emph{all} pairs that satisfy a  sub-trajectory similarity metric. 

\emph{ALSTJ}~\cite{join_subtraj_edbt2020} is the closest work to ours. 
Its similarity metric is based on the spatial span of sub-trajectory pairs that are within a given distance threshold. 
The main differences between ALSTJ and our STS-Join are in three aspects: 
(i)~ALSTJ does not consider the temporal factor and may join trajectories generated at different times,  
while STS-Join requires the trajectories to be close in both space and time so as to be joined; (ii)~ALSTJ may have false negatives in its query result due to its trajectory simplification procedure, while our STS-Join guarantees accurate query results; and 
(iii)~ALSTJ does not consider user privacy while we do. 

\subsection{Trajectory Privacy}\label{sec:related_privacy}
Trajectory privacy has been studied extensively in the last decade~\cite{privacy_dummy, privacy_dummy2,dp_geo_indistinguish, dp_traj_lap,privacy_egemen2, privacy_traj_semantic}. 
For example, a study~\cite{privacy_dummy2} introduces \emph{position dummy} to hide a user’s location by mixing it with fake locations. 
Another study~\cite{k_anonymity_trajectory2} extends \emph{$k$-anonymity} to trajectory data.
Studies~\cite{dp_geo_indistinguish, dp_traj_lap} leverage \emph{differential privacy} for release of trajectory data.
\emph{Geo-indistinguishability}~\cite{dp_geo_indistinguish} generalizes differential privacy to user location data. 
\emph{SDD}~\cite{dp_traj_lap} applies the exponential-based randomized mechanism to trajectory data by sampling a rational distance and direction with noises between locations in the trajectory. 
There are also studies focusing on semantic privacy of trajectories~\cite{privacy_egemen2, privacy_traj_semantic}. Monreale et al.~\cite{privacy_traj_semantic} present a place taxonomy based method to preserve the trajectory semantic privacy. It guarantees that the probability of inferring the sensitive stops of a user is below a threshold. 
Naghizade et al.~\cite{privacy_egemen2} propose an algorithm to protect the semantic information of a trajectory by substituting sensitive stops of a trajectory with less sensitive ones. 

As mentioned earlier, our aim is \emph{not} to propose a new privacy protection scheme but to show that it is feasible to process STS-Join queries with a client-server architecture where the server only stores a modified version of the user trajectories. We use obfuscation for its simplicity. Other privacy protection methods (e.g., position dummy) can be applied in our STS-Join if the modification on the trajectory points can be bounded by a distance threshold, which helps guarantee no false negative query results.
\setlength\tabcolsep{4pt} 
\begin{table}[htp]
\centering
\caption{Frequently Used Symbols}\label{tab:join_symbol}
\begin{tabular}{c|l}
\hline
\textbf{Symbol} & \textbf{Description} \\ \hline \hline
$\mathcal{D}_p$ & An existing trajectory set \\ \hline 
$\mathcal{D}_q$ & A query trajectory set \\ \hline 
$\delta_d$, $\delta_t$ & Trajectory join distance and time thresholds \\ \hline
$\theta_{sp}$ & Simplification threshold \\ \hline
$\theta_{ob}$ & Maximum obfuscated shifting distance \\ \hline
$\mathcal{T}$ & A trajectory \\ \hline
$p_i$ & A point in trajectory \\ \hline
$s_i (\overline{p_i, p_{i+1}})$ & A segment in trajectory \\ \hline
$cdd(s^\mu_i, s^\nu_j)$ & The close-distance duration of $s^\mu_i$ and  $s^\nu_j$  \\ \hline
$cdds(\mathcal{T}_\mu, \mathcal{T}_\nu)$ & The CDD similarity between  $\mathcal{T}_\mu$ and $\mathcal{T}_\nu$ \\ \hline
\end{tabular}
\end{table}

\section{Preliminaries} \label{sec:problemdef}

Given two sets of trajectories  $\mathcal{D}_p$ (\emph{known trajectory set}) and 
$\mathcal{D}_q$ (\emph{query trajectory set}), STS-Join returns pairs
of trajectories $(\mathcal{T}_\mu, \mathcal{T}_\nu) \in \mathcal{D}_p \times \mathcal{D}_q$ with sub-trajectories within $\delta_d$ distance for at least $\delta_t$ time, where $\delta_d$ and $\delta_t$ are query parameters. 
Below, we present a few basic concepts and formulate STS-Join. We list the frequently used symbols in Table~\ref{tab:join_symbol}.


A trajectory $\mathcal{T}$ is formed by a sequence of $|\mathcal{T}|$ sampled points $[p_1, p_2, \ldots, p_{|\mathcal{T}|}]$. 
A point $p_i$ is a triple $\langle x_i, y_i, t_i \rangle$: $p_i$ was generated by a user at location $(x_i,  y_i)$ (in Euclidean space) at time $t_i$. Two adjacent points $p_i$ and $p_{i+1}$ form a segment $s_i = \overline{p_i, p_{i+1}} \in \mathcal{T}$. 


Following previous studies~\cite{join_subtraj_edbt2020, simi_DISSIM}, we consider a constant speed on each trajectory segment. This is valid as real-world trajectories have high sampling rates, e.g., 4.5 seconds in our experiments. The speed may not vary much in such short time frames.
Such a constant-speed setting enables computing a user's location at any time $t \in [t_i, t_{i+1}]$, denoted as 
$(\mathcal{X}_{\mathcal{T}}(t), \mathcal{Y}_{\mathcal{T}}(t))$, 
 given a trajectory $\mathcal{T}$, by linear interpolation:
\begin{equation} \label{eq:XT(t)}
    \big(\mathcal{X}_{\mathcal{T}}(t), \mathcal{Y}_{\mathcal{T}}(t) \big) = \big(x_i+\frac{t-t_i}{t_{i+1}-t_i} (x_{i+1} - x_i), y_i+\frac{t-t_i}{t_{i+1}-t_i} (y_{i+1} - y_i) \big)
\end{equation}
In Fig.~\ref{fig:cdd}, there are two trajectories
$\mathcal{T}_\mu = [p^\mu_1, p^\mu_2, p^\mu_3]$ (the black polyline) and $\mathcal{T}_\nu = [p^\nu_1, p^\nu_2, p^\nu_3]$ (the red polyline).
The solid points in the trajectories represent the sample points, and they are labeled with their timestamps, e.g., $p^\mu_1$ is recorded at 7:00. Using Equation~\ref{eq:XT(t)}, we can derive a user's location on $\mathcal{T}_\mu$ (or $\mathcal{T}_\nu$), e.g., at 7:03, the user should be at $p^\mu_{1'}$. 

Now we can measure the distance between $\mathcal{T}_\mu$ and $\mathcal{T}_\nu$, denoted by $dist(\mathcal{T}_\mu, \mathcal{T}_\nu)$, at any time $t$ 
as the Euclidean distance.
We call this distance the \emph{point distance}.
STS-Join computes the time duration where this distance is within a given threshold $\delta_d$.
\begin{figure}[htp]
  \centering
  \includegraphics[width=0.35\textwidth]{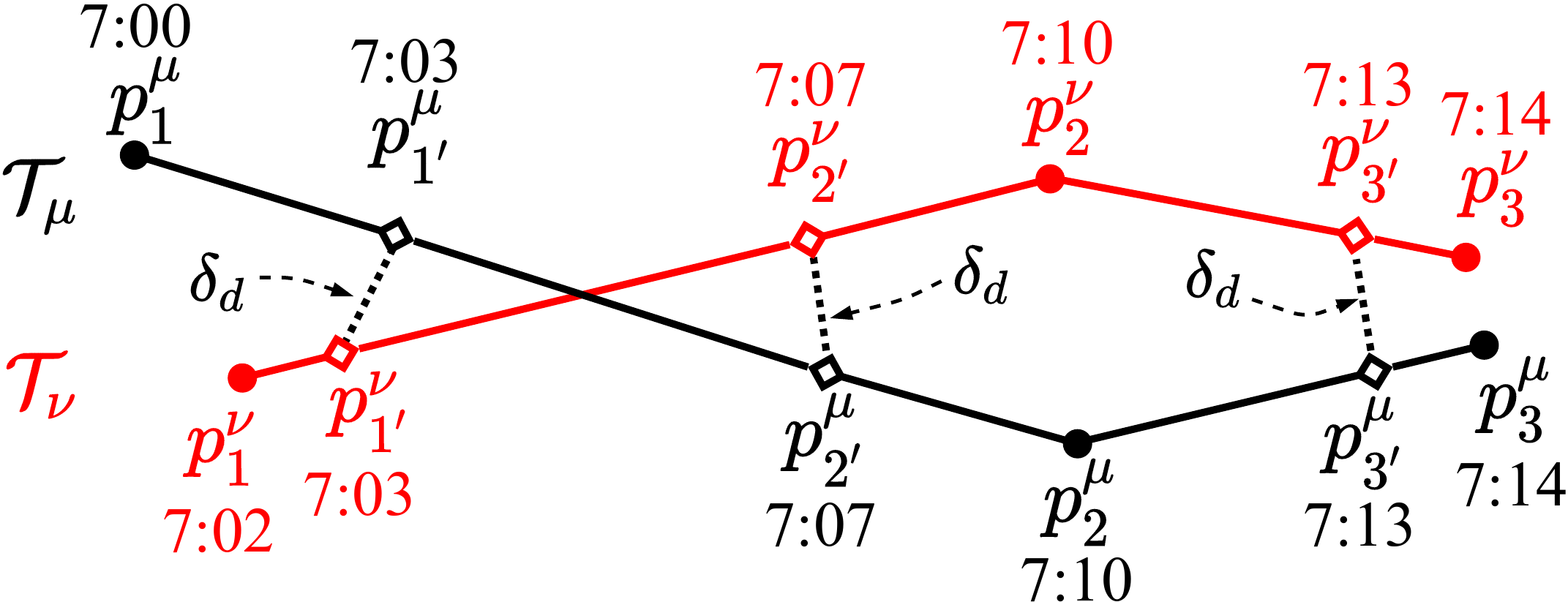}
  \caption{Example of trajectories}\label{fig:cdd}
\end{figure}

Trajectories may be generated at different time spans and with different sampling rates. To compute the point distance between $\mathcal{T}_\mu$ and $\mathcal{T}_\nu$, we need to first define the overlapping time span of segment pairs.
Let $ot(s^\mu_i, s^\nu_j)$ be the overlapping time span between $s^\mu_i$ and $s^\nu_j$, i.e., $[t^\mu_i, t^\mu_{i+1}] \cap [t^\nu_j, t^\nu_{j+1}]$, where $(s^\mu_i, s^\nu_j) \in \mathcal{T}_\mu \times \mathcal{T}_\nu$. Denote the length of the overlapping time span as $|ot(s^\mu_i, s^\nu_j)|$.
Then, in Fig.~\ref{fig:cdd}, $ot(s^\mu_1, s^\nu_1)$ is [7:02, 7:10], and the length is 8 minutes.

\textbf{Close-distance duration.} 
When $ot(s^\mu_i, s^\nu_j) \neq \emptyset$, we call $s^\mu_i$ and $s^\nu_j$ two \emph{time-overlapping} segments. Given such segments, we compute the time range $[t^{\mu\nu}_{i, j\vdash}, t^{\mu\nu}_{i, j\dashv}]$ 
where $dist\big(\mathcal{T}_\mu, \mathcal{T}_\nu) \leqslant \delta_d$, i.e.,
\begin{equation}  \label{eq:sqrt(XT)(t)}
\begin{aligned}
    \sqrt{\big((\mathcal{X}_{\mathcal{T_\mu}}(t) - \mathcal{X}_{\mathcal{T_\nu}}(t)\big)^2 + \big((\mathcal{Y}_{\mathcal{T_\mu}}(t) - \mathcal{Y}_{\mathcal{T_\nu}}(t)\big)^2 } \leqslant \delta_d
\end{aligned}
\end{equation}


To solve this inequality, we expand Equation~\ref{eq:sqrt(XT)(t)} with Equation~\ref{eq:XT(t)} and compute the square of both sides of the inequality:
\begin{equation} \label{eq:distance_to_t}
\begin{aligned}[b]
    \delta_d^2 & \geqslant \big((\mathcal{X}_{\mathcal{T_\mu}}(t) - \mathcal{X}_{\mathcal{T_\nu}}(t)\big)^2 + \big((\mathcal{Y}_{\mathcal{T_\mu}}(t) - \mathcal{Y}_{\mathcal{T_\nu}}(t)\big)^2  \\
    & =  (k_{x}^{\mu} \cdot t + b_{x}^{\mu} - k_{x}^{\nu} \cdot t - b_{x}^{\nu})^2\\
    & \;\;\; + (k_{y}^{\mu}\cdot t + b_{y}^{\mu} - k_{y}^{\nu} \cdot t - b_{y}^{\nu})^2 \\
    & = [(k_{x}^{\mu} - k_{x}^{\nu})^2 + (k_{y}^{\mu} - k_{y}^{\nu})^2]t^2 \\
    & \;\;\; + 2[(k_{x}^{\mu} - k_{x}^{\nu})(b_{x}^{\mu} - b_{x}^{\nu}) 
 +(k_{y}^{\mu} - k_{y}^{\nu})(b_{y}^{\mu} - b_{y}^{\nu})] t \\
    & \;\;\; + (b_{x}^{\mu} - b_{x}^{\nu})^2 + (b_{y}^{\mu} - b_{y}^{\nu})^2, \text{where $t \in ot(s^\mu_i, s^\nu_j)$ and} \\
    &   \begin{cases}
            k_{x}^{\mu} = \frac{x_{i+1}^{\mu} - x_{i}^{\mu}}{t_{i+1}^{\mu} - t_{i}^{\mu}},& b_{x}^{\mu} =\frac{x_{i}^{\mu} t_{i+1}^{\mu} - x_{i+1}^{\mu}t_{i}^{\mu}}{t_{i+1}^{\mu} - t_{i}^{\mu}} \\
            k_{x}^{\nu} = \frac{x_{j+1}^{\nu} - x_{j}^{\nu}}{t_{j+1}^{\nu} - t_{j}^{\nu}},& b_{x}^{\nu} =\frac{x_{j}^{\nu} t_{j+1}^{\nu} - x_{j+1}^{\nu}t_{j}^{\nu}}{t_{j+1}^{\nu} - t_{j}^{\nu}} \\
            k_{y}^{\mu} = \frac{y_{i+1}^{\mu} - y_{i}^{\mu}}{t_{i+1}^{\mu} - t_{i}^{\mu}},& b_{y}^{\mu} =\frac{y_{i}^{\mu} t_{i+1}^{\mu} - y_{i+1}^{\mu}t_{i}^{\mu}}{t_{i+1}^{\mu} - t_{i}^{\mu}} \\
            k_{y}^{\nu} = \frac{y_{j+1}^{\nu} - y_{j}^{\nu}}{t_{j+1}^{\nu} - t_{j}^{\nu}},& b_{y}^{\nu} =\frac{y_{j}^{\nu} t_{j+1}^{\nu} - y_{j+1}^{\nu}t_{j}^{\nu}}{t_{j+1}^{\nu} - t_{j}^{\nu}}
        \end{cases}
\end{aligned}
\end{equation}

The resultant quadratic inequality has just one variable $t$. It can be solved straightforwardly by letting the inequality be equal and computing the roots for the equation with the quadratic formula. We omit the detailed computation for conciseness. 

In Fig.~\ref{fig:cdd}, the distance threshold $\delta_d$ is represented 
by the dotted lines. The distance between the first segments of the two trajectories first decreases and then increases. At time $t =$ 7:03 and 7:07, 
$dist\big(\mathcal{T}_\mu, \mathcal{T}_\nu) = \delta_d$.
which yields the first time range [7:03, 7:07] that satisfies $\delta_d$. 
The distance between the second segments of the two trajectories keeps decreasing, which reaches $\delta_d$ at 7:13. 
This yields the second time range [7:13, 7:14] that satisfies $\delta_d$. 

We call the length of $[t^{\mu\nu}_{i, j\vdash}, t^{\mu\nu}_{i, j\dashv}]$ 
the \emph{close-distance duration} (CDD), denoted as  $cdd(s^\mu_i, s^\nu_j) = t^{\mu\nu}_{i, j\dashv} - t^{\mu\nu}_{i, j\vdash}$.
We define the similarity between two trajectories as their total CDD across all segments. 

\begin{definition}[Close-distance duration similarity]\label{def:cd}
Given a point distance threshold $\delta_d$, the \emph{close-distance duration similarity} (CDDS) of two trajectories $\mathcal{T}_\mu$ and $\mathcal{T}_\nu$, $cdds(\mathcal{T}_\mu, \mathcal{T}_\nu)$, is the sum of 
$cdd(s^\mu_i, s^\nu_j)$ of every time-overlapping segment pair $(s^\mu_i, s^\nu_j) \in \mathcal{T}_\mu \times \mathcal{T}_\nu$.
\begin{equation}
\begin{aligned}
    cdds(\mathcal{T}_\mu, \mathcal{T}_\nu) = \sum_{1 \le i\le|\mathcal{T}_\mu|, 1 \le j\le|\mathcal{T}_\nu|, ot(s^\mu_i, s^\nu_j) \neq \emptyset} {cdd(s^\mu_i, s^\nu_j)}
\end{aligned}
\end{equation}
\end{definition}

CDDS sums up the duration of all close segments including the partial ones. This differs from existing trajectory similarity metrics that require the full  trajectories to be close. 
In Fig.~\ref{fig:cdd}, $cdds(\mathcal{T}_\mu, \mathcal{T}_\nu)$ equals to  the total length of the two time ranges [7:03, 7:07] and [7:13, 7:14], i.e., 5 minutes.

\textbf{Problem definition.} 
Now we can formulate our STS-Join.  
\begin{definition}[STS-Join query]  
Given two trajectory datasets $\mathcal{D}_p$ and $\mathcal{D}_q$, a point distance threshold $\delta_d$, and a close-distance duration similarity threshold $\delta_t$, \emph{STS-Join} returns every trajectory pair 
$(\mathcal{T}_\mu, \mathcal{T}_\nu) \in \mathcal{D}_p \times \mathcal{D}_q$ such that $cdds(\mathcal{T}_\mu, \mathcal{T}_\nu) \geqslant \delta_t$. 
\end{definition}

\section{Index Structure}\label{sec:index}

We assume set $\mathcal{D}_p$ to be known (e.g., user trajectory dumps) and set $\mathcal{D}_q$  to be given at query time (e.g., trajectories of newly confirmed COVID-19 cases). 
We build an index named \emph{STS-Index} over $\mathcal{D}_p$ such that $\mathcal{D}_p$ can be STS-Joined with $\mathcal{D}_q$ efficiently. 
We use a client-server architecture to protect  location privacy. 
On a client, a user's trajectories are stored in their original form, which are obfuscated and sent to the server. The obfuscated trajectories from all clients together are indexed in a tree structure on the server that considers both their spatial and temporal features. 
Next, we detail the index structures on the clients and the server, respectively. 

\subsection{On Client Side}
A user's original trajectories are stored on the client side. We \emph{simplify} and \emph{obfuscate} an original trajectory before sending it (together with the client ID and a local trajectory ID) to the server in order to reduce the communication and protect user's privacy. 
We index the trajectories by their local IDs (e.g., using a sorted array or a B-tree) for fast retrieval at the refinement stage of query processing.

\textbf{Trajectory simplification.} 
First, we simplify an original trajectory $\mathcal{T}_\mu$ by reducing 
the number of sampled points. This reduces the storage space and improves the query efficiency later. Our  simplification algorithm is adapted from the \emph{Douglas–Peucker algorithm}~(DP)~\cite{simp_dp}. 
The native DP algorithm ignores the temporal dimension. Consider two sampled points $p^\mu_i$ and $p^\mu_{i+k} (k > 1)$ on $\mathcal{T}_\mu$. For all other sampled points between $p^\mu_i$  and $p^\mu_{i+k}$, if their perpendicular distances to the segment between $p^\mu_i$ and $p^\mu_{i+k}$ are within a simplification threshold $\theta_{sp}$ (an empirical parameter), then these points are all removed from $\mathcal{T}_\mu$ by DP. 

In our case, since we interpolate user locations on the trajectory segments, we require the user location on the simplified segment $\overline{p^\mu_i, p^\mu_{i+k}}$ to be within $\theta_{sp}$ distance from that on the original segments at every time point $t \in [t^\mu_i, t^\mu_{i+k}]$.  This guarantees no false negatives in STS-Join over the simplified trajectories. 
Fig.~\ref{fig:simplification} shows an example. The original (black) trajectory $\mathcal{T}_\mu$  has three segments, which is simplified to just one (red) segment  $\overline{{p^{\widetilde{\mu}}_1},{p^{\widetilde{\mu}}_4}}$. 
We compute $p^{\widetilde{\mu}}_2$ and $p^{\widetilde{\mu}}_3$ on $\overline{{p^{\widetilde{\mu}}_1},{p^{\widetilde{\mu}}_4}}$ at times $t^\mu_2$ and $t^\mu_3$  (i.e., the time points of $p^\mu_2$ and $p^\mu_3$), respectively. 
To ensure valid simplification, the distance between $p^\mu_2$ and $p^{\widetilde{\mu}}_2$ (i.e., $d_2$) and that between 
$p^\mu_3$ and $p^{\widetilde{\mu}}_3$ (i.e., $d_3$) must both be within $\theta_{sp}$. 
This contrasts to the native DP that examines the perpendicular distances of 
$p^\mu_2$ and $p^\mu_3$ (i.e., $d^\perp_2$ and $d^\perp_3$), which are shorter and may lead to false negatives at query processing.  

\begin{figure}[htp]
  \centering
  \includegraphics[width=0.27\textwidth]{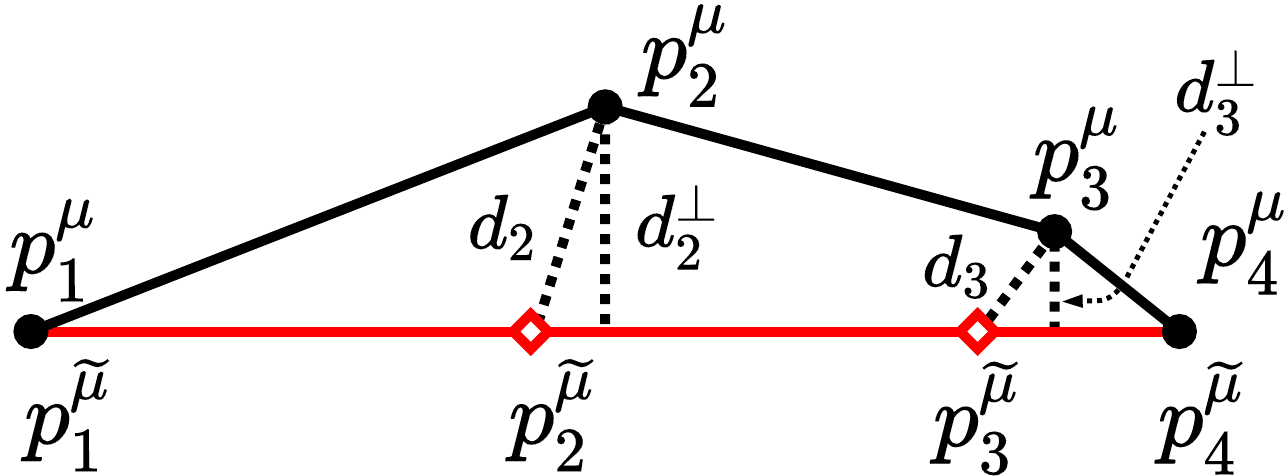}
  \caption{Example of trajectory simplification}\label{fig:simplification}
\end{figure}

\textbf{Trajectory obfuscation.} 
Our simplified trajectories retain a subset of the trajectory points. 
To protect privacy, we further adapt the \emph{bounded Laplace mechanism} (i.e., an algorithm) to obfuscate the simplified trajectories  as inspired by previous studies~\cite{dp_first, dp_bounded_laplace}.

The bounded Laplace mechanism adds a noise from the bounded Laplace distribution to the output of a (query) function. In our problem, we add a bounded noise to each dimension of a point on a simplified trajectory to protect users' location privacy.
The probability distribution function (PDF) of the added noise is~\cite{dp_bounded_laplace}:
\begin{equation} \label{eq:bounded_lap_pdf}
\begin{aligned}
    Pr(x) = \lambda \cdot \frac{1}{2b} \cdot \exp(\frac{-|x|}{b}), x \in [-\theta_{ob}, \theta_{ob}]
\end{aligned}
\end{equation}
where $\lambda$ is a constant, $b$ is the bias of the distribution, $\theta_{ob} \in \mathbb{R}^+$ is the obfuscated distance threshold.
Since the domain of the distribution is bounded in $[-\theta_{ob}, \theta_{ob}]$, the integral of PDF should be 1 for $x \in [-\theta_{ob}, \theta_{ob}]$. 
This yields the value of $\lambda$: 
\begin{equation}\label{eq:lambda}
\begin{aligned}
\lambda
    & = (\int_{-\theta_{ob}}^{\theta_{ob}} \frac{1}{2b} \cdot \exp(\frac{-|x|}{b}) dx)^{-1} = (1 - \exp(-\frac{\theta_{ob}}{b}))^{-1}
\end{aligned}
\end{equation}
We then leverage the \emph{inverse cumulative distribution function} to generate random noises that satisfy the PDF. 
Firstly, we derive the CDF from the PDF by computing the integral of the PDF for $x<0$ and $x\geqslant0$ respectively:
\begin{equation}\label{eq:CDF}
\begin{aligned}
    F(x) = \begin{cases}
        \frac{\lambda}{2} \cdot (\exp(\frac{x}{b})-\exp(\frac{-\theta_{ob}}{b})), & x <0 \\
        \frac{1}{2}+\frac{\lambda}{2}\cdot(1-\exp(\frac{-x}{b})), & x \geqslant 0
    \end{cases}
\end{aligned}
\end{equation}
Then, we can obtain the inverse CDF from Equation~\ref{eq:CDF}:
\begin{equation}
    x = \begin{cases}
        -b \cdot \ln(1 + \lambda^{-1} - 2\lambda^{-1}y), & y \in (0, 0.5] \\
         b \cdot\ln(1 - \lambda^{-1} + 2\lambda^{-1}y), & y \in (0.5, 1)
    \end{cases}
\end{equation}
Here, variable $y$ is a random number in $(0, 1)$. We generate $y$ and use it to obtain noise $x$, which is then added to each point coordinate of the simplified trajectory. Parameter $\lambda$ is determined by $b$ and $\theta_{ob}$ (Equation~\ref{eq:lambda}), while $b$ is the distribution bias.

\subsection{On Server Side} \label{sec:build_index}
The obfuscated trajectories are stored in a tree structure on the server that indexes both the spatial and temporal features of the trajectories (in the form of segments). 
As Fig.~\ref{fig:index_structure} shows, the top levels of the structure together can be seen as a B-tree  that indexes time intervals hierarchically, and the node capacity (which is a system parameter) of the example is 20. The time span of an entry in the leaf nodes is 3 minutes which is determined by the time span of the trajectory segments, and every entry points to a spatial index for the trajectory segments in that interval. If a segment spans across the intervals of multiple entries (which is rare as the segments are usually short even after simplification), we break the segment into multiple segments to suit the entry intervals.
For example, if a segment starts at 8:29:50 and ends at 8:30:03, we break the segment at 8:30:00 into two segments that suit the time intervals [8:27:00, 8:30:00) (i.e., $N_5$) and [8:30:00, 8:33:00) of the leaf entries.

\begin{figure}[htp]
  \centering
  \includegraphics[width=0.47\textwidth]{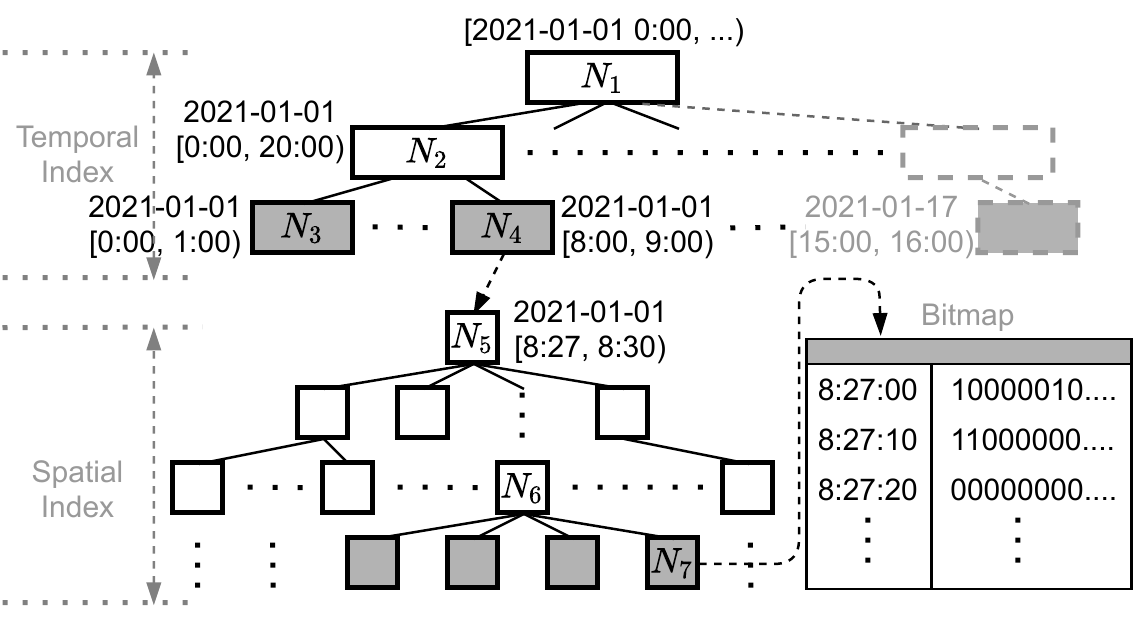}
  \caption{STS-Index server-side structure}\label{fig:index_structure}
\end{figure}

We create a \emph{quasi-quadtree} as the spatial index to obtain nodes with non-overlapping \emph{minimum bounding rectangles}~(MBR). We use endpoints of the segments to build this tree. First, we insert the endpoints into the root. Once the root node capacity is reached, we partition the space into four quadrants each forming a child node. Every  endpoint is moved to the child node enclosing it, while the root node now stores pointers to the child nodes. The process is done recursively until every node is within its capacity. In a leaf node, we further store the segments overlapped by the node MBR (a segment may be in multiple nodes), together with the IDs of the corresponding trajectories and clients. 

When the data distribution is skewed, we may have many segments in the same leaf node. We add bitmaps to help query such segments.
The bitmaps correspond to disjoint intervals that together cover the time interval of the segments in the node.
The number of bitmaps can be empirically  determined based on the sampling rate. 
Each bit value denotes whether a segment overlaps with a time interval. 
In Fig.~\ref{fig:index_structure}, 
we build a bitmap for $N_7$ with 10-second intervals, where  both the first and the seventh segments in $N_7$ overlap with the interval [8:27:00, 8:27:10).


\textbf{Update handling.} 
When there are new obfuscated trajectories, their segments (and the endpoints of the segments) are added to the index by a top-down traversal to find the nodes to be inserted into (following the query procedure in the next section). 
Trajectory deletion can be also done by first a tree traversal to locate the trajectory segments to be deleted. Then, the segments are removed, together with any empty nodes.  

\section{STS-Join Query Processing} \label{sec:traj_join}
We illustrate how to join a query trajectory set $\mathcal{D}_q$ with a known trajectory set $\mathcal{D}_p$ indexed in  
our STS-Index. 
Set $\mathcal{D}_q$ is not obfuscated, e.g., the trajectories of confirmed COVID-19 cases are reported to the authority for contact tracing. 
Firstly, we present our STS-Join query algorithm and optimize it with a backtracking technique.
Then, we illustrate how to prune dissimilar pairs on the server by bounding the actual trajectory similarity based on obfuscated trajectories. 
Finally, we analyze the algorithm cost.  

\subsection{The STS-Join Query Algorithm} \label{sec:traj_join_algo}
A straightforward algorithm is to query every segment from $\mathcal{D}_q$ independently over STS-Index, identify the data segments satisfying the query, and send the query trajectory to the corresponding clients for verification against the original trajectories. 

We observe that, segments of a trajectory are connected end to end. Points of adjacent segments are likely to be in the leaf nodes that are in short-hop distance from each other in the index. Besides, our quasi-quadtrees divide the space without overlaps. 
By leveraging these features, we propose an efficient backtracking-based join algorithm that starts from the inner nodes instead of the root for querying each segment of a trajectory.

\textbf{MBR expansion for accurate query processing.} To ensure no false dismissals, we expand the MBRs of query segments to cover the simplified-and-obfuscated trajectories, as well as the query distance threshold $\delta_d$. Overall, we expand the MBR of a query segment by $\delta_d+\theta_{sp}+\theta_{ob}$ on each side.
This ensures no false negatives, because by definition, the distance between a point on an original trajectory segment and the  corresponding point on its simplified-and-obfuscated version cannot be greater than $\delta_d+\theta_{sp}+\theta_{ob}$ in any dimension.
We use $\lceil{mbr}\rceil_{s_{i}}$ to denote the expanded MBR of $s_i$. 
In Fig.~\ref{fig:pivot}, $\lceil{mbr}\rceil_{s^{\hat{\nu}}_{1}}$ is the expanded MBR of the segment  $s^{\hat{\nu}}_{1}$.

\begin{figure}[htp]
  \centering
  \includegraphics[width=0.3\textwidth]{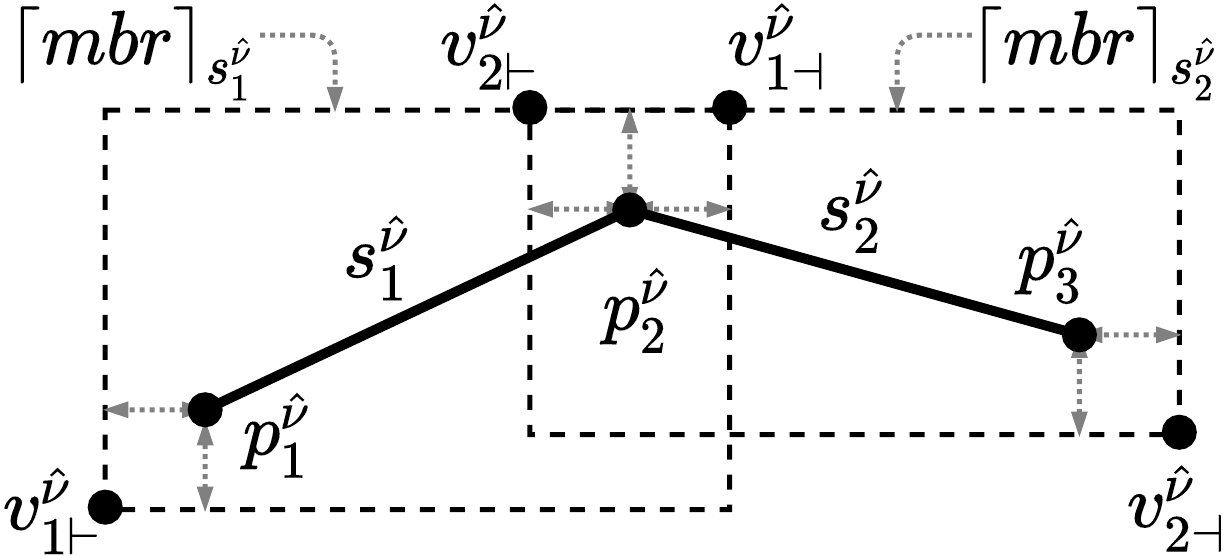}
  \caption{Example of expanded MBRs}\label{fig:pivot}
\end{figure}

\textbf{Pivots.} We use \emph{pivots} to help locate the adjacent segments in our algorithm. Pivots are vertices of the expanded MBR of a query segment that are close to the endpoints of the segment. Using a pivot instead of an endpoint helps avoid false negative query results, since the query MBR is expanded. 
In Fig.~\ref{fig:pivot}, $v^{\hat{\nu}}_{1\vdash}$ and $v^{\hat{\nu}}_{1\dashv}$ are the pivots of query segment $s^{\hat{\nu}}_{1}$.

\begin{algorithm}
\caption{STS-Join (server side)} \label{algo:backtrace_join}
\KwIn{$\mathcal{D}_q$: query trajectory set; $\delta_d$: query distance threshold} 
\KwOut{$\mathcal{P}$: the set of STS-Joined trajectory pairs}
\BlankLine
\For {$\mathcal{T}_\nu$ in $\mathcal{D}_q$} {
    $\mathcal{T}_{\hat{\nu}} \leftarrow$ split $\mathcal{T}_\nu$ to suit node intervals of the temporal index\;
    \For {$s^{\hat{\nu}}_i$ in $\mathcal{T}_{\hat{\nu}}$} {
        \If {$i$ \emph{is 1} \textbf{or} $s^{\hat{\nu}}_{i}$ \emph{and} $s^{\hat{\nu}}_{i-1}$ \emph{are in different quasi-quadtrees}} {
            $pivot \leftarrow v^{\hat{\nu}}_{i\vdash}$\;
            $N \leftarrow$ the root of the quasi-quadtree whose time interval overlaps with that of $s^{\hat{\nu}}_i$\; 
            $N \leftarrow FindNode(N, pivot)$\;
        }
        $pivot \leftarrow v^{\hat{\nu}}_{{i+1}\vdash}$\; 
        $N_p \leftarrow Backtrack(N, v^{\hat{\nu}}_{i\dashv})$\;
        $Q.enqueue(N_p)$\;
        \While {$Q \neq \emptyset$} {
            $N_p \leftarrow Q.dequeue()$\;
            \For {each $entry$ in $N_p$} {
                 \If {$overlap(entry.mbr, \ceil{mbr}_{s^{\hat{\nu}}_i})$} {
                    \If {$N_p$ \emph{is not a leaf node}} {
                        $Q.enqueue(entry.child)$\;
                    }
                    \Else {
                        add $\langle s^{\widetilde{\mu}}, s^{\hat{\nu}}_i \rangle$ into $\mathcal{S}$, where $s^\mu$ is the segment indexed at $entry$\;
                    }
                }
            }
            \If {$N_p$ \emph{is a leaf node} \textbf{and} \emph{its} $mbr$ \emph{covers} $pivot$ } {
                $N \leftarrow N_p$\; 
            }
        }
    }
}
    Update every pair $\langle s^{\widetilde{\mu}}, s^{\hat{\nu}}_i \rangle \in \mathcal{S}$ to its corresponding non-time-interval-split segment pair $\langle s^{\widetilde{\mu}}, s^\nu \rangle$\;
    Group the pairs in $\mathcal{S}$ by the client ID of $s^{\widetilde{\mu}}$, and send corresponding pairs to the clients for further verification\;
    $\mathcal{P} \leftarrow $ the set of trajectory pairs that are returned from clients\;
    \Return $\mathcal{P}$\;
\end{algorithm}


\textbf{On server side.} As summarized in Algorithm \ref{algo:backtrace_join},  our server-side algorithm iteratively searches for similar segments for the query trajectories (lines 1 to 20). 
For every query trajectory $\mathcal{T}_\nu$, we split its segments according to the time intervals of the root nodes of the quasi-quadtrees (line 2), like we did in index construction. 
Then, we search for data segments similar to every segment in the split query trajectory  $\mathcal{T}_{\hat{\nu}}$ (lines 3 to 20).
We first update $pivot$ and its corresponding leaf node $N$. This is only needed for the first segment in $\mathcal{T}_{\hat{\nu}}$ or when we move on to a  segment in a new time interval (lines 4 to 7), which rarely happens.
Function $FindNode(N, pivot)$ locates node $N$ whose MBR covers the pivot in the quasi-quadtree by a point query (line 7). 
Then, we can leverage node $N$ to backtrack with 
function $Backtrack(N, p)$ that starts from $N$ and recursively traces back to the ancestor node whose MBR covers $p$. For each query segment, we stop the backtracking at the ancestor node $N_p$ that covers the upper pivot $v^{\hat{\nu}}_{i\dashv}$ of the expanded MBR of the current query segment (lines 8 and 9). Then, we search for similar segment pairs from $N_p$ for the current query segment (lines 10 to~20). 
For pruning, only tree nodes whose MBRs overlap with the expanded MBR $\ceil{mbr}_{s^{\hat{\nu}}_i}$ of the query segment are visited, which are stored in a queue $Q$ to support the traversal (lines 14 to 16). When the traversal reaches a simplified-and-obfuscated segment $s^{\widetilde{\mu}}$, we add $\langle s^{\widetilde{\mu}}, s^{\hat{\nu}}_i \rangle$ to the result set $\mathcal{S}$ (line 18).
Meanwhile, we verify each leaf node for whether it contains the lower pivot of the next query segment which will be used at the stage of backtracking in the next segment query (line 19).
After all query trajectories have been processed, we update the query segments in $\mathcal{S}$ to their corresponding non-time-interval-split segments (line 21).
Then, we group the pairs in $\mathcal{S}$ by the client that generated $s^{\widetilde{\mu}}$, and send the pairs to the clients based on the client IDs of the segments (line 22). 
 STS-Join verifies the trajectory similarity on the clients, since the server only stores  simplified-and-obfuscated trajectories.
After each client has computed the trajectory similarity, it returns the result set to the server. 

\textbf{On client side.}
On each client, the trajectories corresponding to the segments  $s^{\widetilde{\mu}}$ received from the server are retrieved (by ID lookups using the local trajectory IDs of the segments). Then, we compute the CDD of the segment pairs from the server and add up the CDDS for every trajectory pair formed by the segment pairs. The pairs satisfying the time threshold $\delta_t$ are returned as the query result to the server. This  guarantees no false positives. 

\textbf{Discussion.} 
Backtracking is not limited to quasi-quadtrees. It is applicable to all space-partitioning indices in which the process can stop at a common parent node. A query starting from such parent nodes will not have false negatives, since there are no overlaps among MBRs on the same level in the tree. That means one location can be covered only by one MBR at each level. 
In addition, backtracking can be applied in index construction, insertion, and deletion, because segments are inserted or deleted sequentially, and we can leverage the common parent node approach to reduce the node accesses when operating on the next segment.

\vspace{-2mm}
\subsection{CDDS-Based Pruning}\label{subsec:sim_bound}
Algorithm~\ref{algo:backtrace_join} sends all segment pairs that may satisfy the query distance threshold $\delta_d$ to the clients for further verification and CDDS computation. In this subsection, we compute an upper bound of the actual CDDS between a query trajectory $\mathcal{T}_\nu$ and a known trajectory $\mathcal{T}_\mu$ using $\mathcal{T}_\nu$ and the simplified-and-obfuscated  segments of $\mathcal{T}_\mu$ stored on the server. We only send the segment pairs to the corresponding client when the upper bound exceeds $\delta_t$, thus saving both communication costs between the server and the clients and computation costs on the clients. This essentially adds a   subprocedure to prune the segment pairs before Line 22 of Algorithm~\ref{algo:backtrace_join} using an upper bound of $cdds(\mathcal{T}_\mu, \mathcal{T}_\nu)$. For simplicity, we only describe the CDDS-based pruning procedure below but do not repeat the full pseudocode of the STS-Join algorithm powered by it. 

\textbf{CDDS-based pruning procedure.} 
We group the segment pairs that satisfy the query distance threshold (i.e., the segment pairs in set $\mathcal{S}$ at  Line 21 of Algorithm~\ref{algo:backtrace_join}) by their client IDs, local trajectory IDs, and query trajectory IDs. The segment pairs that share the same client ID, local trajectory ID, and query trajectory ID all come from the same pair of simplified-and-obfuscated known and query trajectories $(\mathcal{T}_{\widetilde{\mu}}, \mathcal{T}_\nu)$ for CDDS upper bound computation. Let the set of such segment pairs be $\mathcal{S}_{\widetilde{\mu}, \nu}$ and the original known trajectory corresponding to $\mathcal{T}_{\widetilde{\mu}}$ be $\mathcal{T}_\mu$, respectively. Then, we compute an upper bound of the 
close-distance duration (i.e., CDD) for every pair of segments in $\mathcal{S}_{\widetilde{\mu}, \nu}$. Summing up these upper bounds for all segment pairs in $\mathcal{S}_{\widetilde{\mu}, \nu}$ yields our upper bound of $cdds(\mathcal{T}_\mu, \mathcal{T}_\nu)$, since these pairs are the only ones in $\mathcal{T}_{\widetilde{\mu}} \times \mathcal{T}_\nu$ (and hence $\mathcal{T}_{\mu} \times \mathcal{T}_\nu$)  that satisfy the query distance threshold by definition. The set $\mathcal{S}_{\widetilde{\mu}, \nu}$ is pruned from being sent to a client if the CDDS upper bound is less than $\delta_t$. Next, we detail our CDD upper bound computation. 

\textbf{CDD upper bound.} Recall that the CDD between a known segment $s^\mu_i$ of $\mathcal{T}_{\mu}$ and a query segment $s^\nu_j$ of $\mathcal{T}_\nu$, $cdd(s^\mu_i, s^\mu_j)$, is the time duration when  their spatial distance is within $\delta_d$. We further define the CDD between a simplified-and-obfuscated segment $s^{\widetilde{\mu}}_i$ and a query segment $s^\nu_j$ as 
the time duration where their spatial distance is within $\delta_d + \theta_{sp} + \sqrt{2}\theta_{ob}$, denoted as  $\overline{cdd}(s^{\widetilde{\mu}}_i, s^\nu_j)$. 
Then, for all known segments $s^\mu_i, s^\mu_{i+1}, \ldots, s^\mu_{i+\Delta}$  corresponding to $s^{\widetilde{\mu}}_i$, $\sum_{k=0}^
\Delta cdd(s^\mu_{i+k}, s^\mu_j) \le \overline{cdd}(s^{\widetilde{\mu}}_i, s^\nu_j)$. This is because, given a point on a known segment, its corresponding point on the simplified-and-obfuscated segment is within a distance of $\theta_{sp} + \sqrt{2}\theta_{ob}$ by definition. Thus, if the distance between (points on) $s^\mu_{i+k}$ and $s^\nu_j$ is within $\delta_d$, the distance between (points on) $s^{\widetilde{\mu}}_i$ and $s^\nu_j$ must be within $\delta_d + \theta_{sp} + \sqrt{2}\theta_{ob}$. 
Therefore, we use $\overline{cdd}(s^{\widetilde{\mu}}_i, s^\nu_j)$ as our CDD upper bound of the original known trajectories. 

We formulate the CDD upper bound and show its correctness with the following lemma. 

\vspace{-3mm}
\begin{figure}[t]
  \centering
  \includegraphics[width=0.34\textwidth]{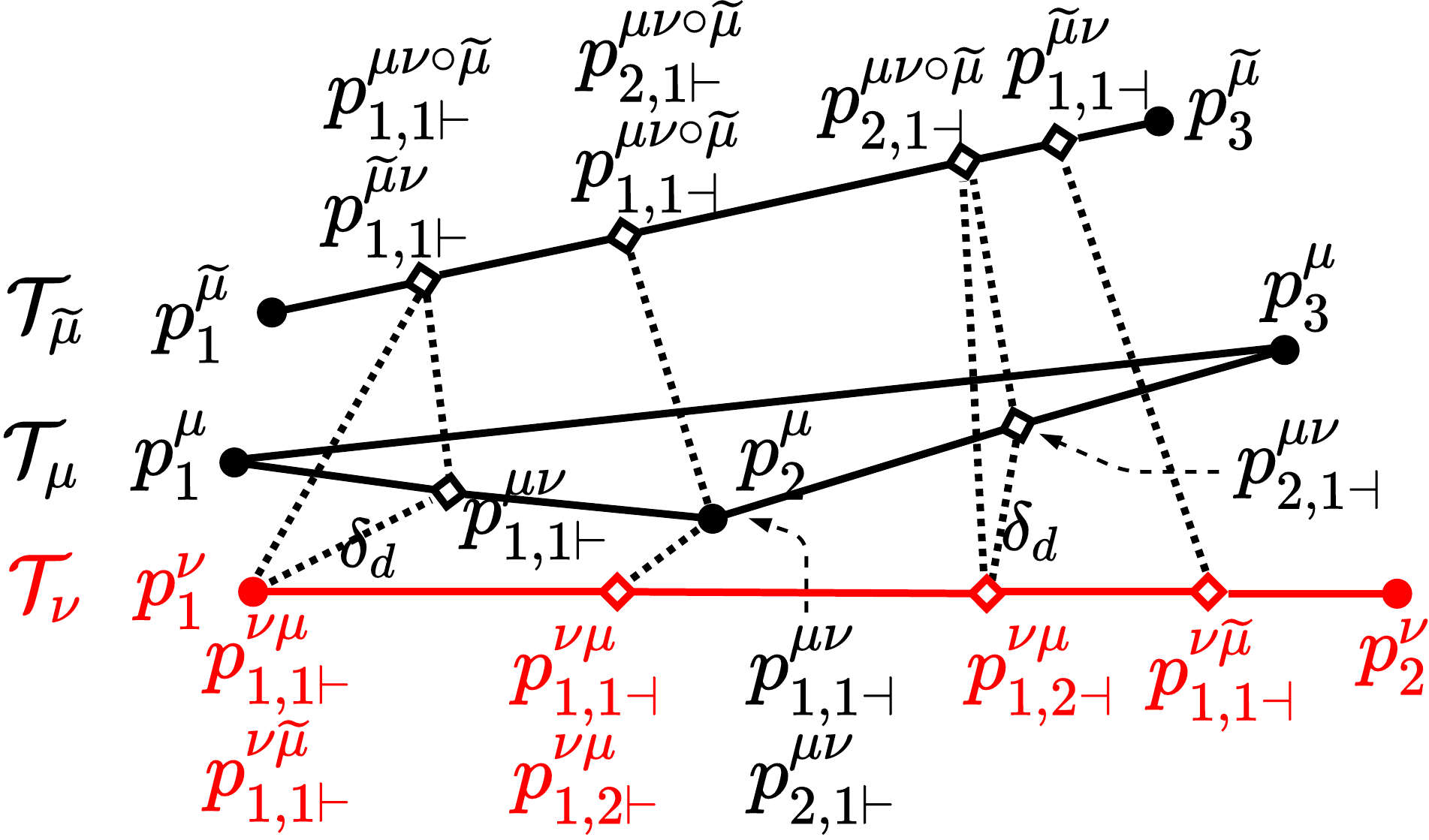}
  \caption{Example of trajectory similarity bounding}\label{fig:bound}
  \vspace{-1mm}
\end{figure}

\begin{lemma} \label{lemma:bound_seg}
Given a query segment $s^\nu_j \in \mathcal{T}_\nu$, a sequence of known segments $s^\mu_i, s^\mu_{i+1}, \ldots,  s^\mu_{i+\Delta} \in \mathcal{T_\mu}$, and their corresponding simplified-and-obfuscated segment $s^{\widetilde{\mu}}_i \in \mathcal{T}_{\widetilde{\mu}}$, we have:
\begin{equation}
\sum\nolimits_{k=0}^{\Delta} cdd(s^\mu_{i+k}, s^\nu_j) \leqslant \overline{cdd}(s^{\widetilde{\mu}}_i, s^\nu_j) 
\end{equation}
where $\overline{cdd}$ is the CDD with distance threshold $\delta_d + \theta_{sp} + \sqrt{2}\theta_{ob}$.
\end{lemma}

\begin{proof}
We use Fig.~\ref{fig:bound} to help illustrate the proof, where known segments $s^\mu_{i}$ can be  $\overline{p^\mu_1,p^\mu_2}$ and $\overline{p^\mu_2,p^\mu_3}$ correspond to a simplified-and-obfuscated segment $s^{\widetilde{\mu}}_{i} = \overline{p^{\widetilde{\mu}}_1,p^{\widetilde{\mu}}_3}$, and the query segment is shown as $s^\nu_{j} = \overline{p^\nu_1,p^\nu_2}$. 
Besides, any two points connected by a dash line have the same timestamp. 

Given a known segment $s^\mu_{i+k}$  and a query segment $s^\nu_j$, their CDD can be non-zero only if they overlap in their time span, i.e.,  $ot(s^\mu_{i+k}, s^\nu_j) \neq \emptyset$. Further, if the CDD is non-zero, it is defined by the two roots (i.e., two time points) of the quadratic equation $dist^2(s^\mu_{i+k}, s^\nu_j) = {\delta^2_d}$ (cf. Equation~\ref{eq:distance_to_t}).  
Let the two roots be $t^{\nu\mu}_{j,i+k\vdash}$ and $t^{\nu\mu}_{j,i+k\dashv}$. 
Then, $cdd(s^\mu_{i+k}, s^\nu_j) = t^{\nu\mu}_{j,i+k\dashv} - t^{\nu\mu}_{j,i+k\vdash}$. 
Note that, if either $t^{\nu\mu}_{j,i+k\vdash}$ or $t^{\nu\mu}_{j,i+k\dashv}$ is outside the range of  $ot(s^\mu_{i+k}, s^\nu_j)$, we replace it with the corresponding boundary value of $ot(s^\mu_{i+k}, s^\nu_j)$ to meet the overlapping time span requirement of the segments. Let the points corresponding to  $t^{\nu\mu}_{j,i+k\vdash}$ on $s^\mu_{i+k}$ and  $s^\nu_{j}$ be $p^{\mu\nu}_{i+k,j\vdash}$ and $p^{\nu\mu}_{j,i+k\vdash}$, respectively. 
Similarly, let the points corresponding to 
$t^{\nu\mu}_{j,i+k\dashv}$ on the two segments be $p^{\nu\mu}_{j,i+k\dashv}$ and $p^{\mu\nu}_{i+k,j\dashv}$, respectively. In Fig.~\ref{fig:bound}, these four points on $s^\mu_{i+k}$ and $s^\nu_j$ are $p^{\nu\mu}_{1,1\vdash}$, $p^{\mu\nu}_{1,1\vdash}$, $p^{\nu\mu}_{1,1\dashv}$, and $p^{\mu\nu}_{1,1\dashv}$ ($i = 1$, $j = 1$, and $k = 0$).

Then, we analyze the distance between a known segment $s^\mu_{i+k}$ and its corresponding simplified-and-obfuscated segment $s^{\widetilde{\mu}}_{i}$. 
Note that a sequence of known segments can be simplified into a single segment with a simplification threshold $\theta_{sp}$, while the simplified segment is further  obfuscated by shifting the endpoints with a maximum shifting distance $\theta_{ob}$ along each dimension.
Thus, the distance between any point $p$ on $s^\mu_{i+k}$ and its corresponding point (i.e., the point at the same time $t$ as that of $p$) on $s^{\widetilde{\mu}}_{i}$, is bounded within $\theta_{sp} + \sqrt{2}\theta_{ob}$.
This applies to the distance between points (at any given time $t$) on $\overline{p^\mu_1,p^\mu_2}$ ($\overline{p^\mu_2,p^\mu_3}$) and $\overline{p^{\widetilde{\mu}}_1,p^{\widetilde{\mu}}_3}$ in Fig.~\ref{fig:bound}.

Next, we derive the distance between the query segment and the simplified-and-obfuscated segment by leveraging the distance relationship above. 
By definition, the distance between the known segment  $s^\mu_{i+k}$ and the query segment $s^\nu_j$ does not exceed $\delta_d$ when $t \in [t^{\nu\mu}_{j,i+k\vdash}, t^{\nu\mu}_{j,i+k\dashv}]$, while 
the distance between $s^\mu_{i+k}$ and its corresponding simplified-and-obfuscated segment $s^{\widetilde{\mu}}_{i+k}$ does not exceed $\theta_{sp} + \sqrt{2}\theta_{ob}$ 
Therefore, the distance $dist(s^{\widetilde{\mu}}_{i+k}, s^\nu_j)$ between the query segment and the simplified-and-obfuscated known segment does not exceed  $\delta_d + \theta_{sp} + \sqrt{2}\theta_{ob}$, when $t \in [t^{\nu\mu}_{j,i+k\vdash}, t^{\nu\mu}_{j,i+k\dashv}]$. 
In Fig.~\ref{fig:bound}, we locate a point $p^{\mu\nu
\circ\widetilde{\mu}}_{2,1\dashv}$ on $\overline{p^{\widetilde{\mu}}_1,p^{\widetilde{\mu}}_3}$ whose timestamp is $t^{\nu\mu}_{1,2\dashv}$ ($i = 1$, $j = 1$, and $k = 1$).
Then, the distance between $p^{\mu\nu\circ\widetilde{\mu}}_{2,1\dashv}$ and $p^{\mu\nu}_{2,1\dashv}$ does not exceed $\theta_{sp} + \sqrt{2}\theta_{ob}$, 
while the distance between $p^{\mu\nu}_{2,1\dashv}$ and $p^{\nu\mu}_{1,2\dashv}$ is $\delta_d$ as shown above. Thus, the distance between  $p^{\mu\nu\circ\widetilde{\mu}}_{2,1\dashv}$ and $p^{\nu\mu}_{1,2\dashv}$ is within $\delta_d + \theta_{sp} + \sqrt{2}\theta_{ob}$.

Since there exists a time range $[t^{\nu\mu}_{j,i+k\vdash}, t^{\nu\mu}_{j,i+k\dashv}]$ where the distance between $s^{\widetilde{\mu}}_{i+k}$ and $s^\nu_j$ does not exceed $\delta_d + \theta_{sp} + \sqrt{2}\theta_{ob}$, 
the quadratic distance equation  
$dist^2(s^{\widetilde{\mu}}_{i+k}, s^\nu_j) = {(\delta_d + \theta_{sp} + \sqrt{2}\theta_{ob})}^2$ must have two roots, which are denoted as $t^{\nu\widetilde{\mu}}_{j,i+k\vdash}$ and $t^{\nu\widetilde{\mu}}_{j,i+k\dashv}$, respectively.
Also, since this quadratic function
has a non-negative quadratic coefficient (cf. Equation~\ref{eq:distance_to_t}), and the distance does not exceed $\delta_d + \theta_{sp} + \sqrt{2}\theta_{ob}$ when $t \in [t^{\nu\mu}_{j,i+k\vdash}, t^{\nu\mu}_{j,i+k\dashv}]$, we can derive that $t^{\nu\widetilde{\mu}}_{j,i+k\vdash} \leqslant t^{\nu\mu}_{j,i+k\vdash} < t^{\nu\mu}_{j,i+k\dashv} \leqslant t^{\nu\widetilde{\mu}}_{j,i+k\dashv}$.
Since $cdd(s^\mu_{i+k}, s^\nu_j) = t^{\nu\mu}_{j,i+k\dashv} - t^{\nu\mu}_{j,i+k\vdash}$ and $\overline{cdd}(s^{\widetilde{\mu}}_{i+k}, s^\nu_j) = t^{\nu\widetilde{\mu}}_{j,i+k\dashv} - t^{\nu\widetilde{\mu}}_{j,i+k\vdash}$, we have $cdd(s^\mu_{i+k}, s^\nu_j) \leqslant \overline{cdd}(s^{\widetilde{\mu}}_{i+k}, s^\nu_j)$ when $t \in ot(s^\mu_{i+k}, s^\nu_j)$. Recall that  $\overline{cdd}$ is CDD computed with distance threshold $\delta_d + \theta_{sp} + \sqrt{2}\theta_{ob}$.
In Fig.~\ref{fig:bound}, $t^{\nu{\widetilde{\mu}}}_{1,1\dashv}$ is the upper boundary of $\overline{cdd}(\overline{p^{\widetilde{\mu}}_1,p^{\widetilde{\mu}}_3}, \overline{p^{\nu}_1,p^{\nu}_2})$, and the corresponding points on the query segment and the simplified-and-obfuscated  segment are $p^{\nu{\widetilde{\mu}}}_{1,1\dashv}$ and $p^{{\widetilde{\mu}}\nu}_{1,1\dashv}$ ($i = 1$, $j = 1$, and $k = 1$) respectively.
Meanwhile, when $t = t^{\nu\mu}_{1,2\dashv}$, the distance between $\overline{p^{\widetilde{\mu}}_1, p^{\widetilde{\mu}}_3}$ and $\overline{p^{\nu}_1,p^{\nu}_2}$ does not exceed  $\delta_d + \theta_{sp} + \sqrt{2}\theta_{ob}$. Thus, we  have $t^{\nu\mu}_{1,2\dashv} \leqslant t^{\nu{\widetilde{\mu}}}_{1,1\dashv}$.
Such an inequality is also applicable to other CDD boundaries. 


Finally, we sum up the CDD of each segment in $s^\mu_i, s^\mu_{i+1}, \ldots, s^\mu_{i+\Delta}$ with $s^\nu_j$, and we sum up $\overline{cdd}(s^{\widetilde{\mu}}_i, s^\nu_j)$ that corresponds to different time ranges in $ot(s^\mu_{i+k}, s^\nu_j)$ where $0 \leqslant k \leqslant \Delta$. Since the inequality is satisfied on each separate time range, such inequality is also satisfied for the sums, i.e., $cdd(\overline{p^\mu_1,p^\mu_2}, \overline{p^\nu_1,p^\nu_2})$ + $cdd(\overline{p^\mu_2,p^\mu_3}, \overline{p^\nu_1,p^\nu_2}) \leqslant \overline{cdd}(\overline{p^{\widetilde{\mu}}_1,p^{\widetilde{\mu}}_3}, \overline{p^\nu_1,p^\nu_2})$ in Fig.~\ref{fig:bound}. This completes the proof.
\end{proof}

Given Lemma~\ref{lemma:bound_seg}, we have the following lemma to bound the CDDS for the original known trajectories. 
\begin{lemma} \label{lemma:bound_traj}
Given a query trajectory $\mathcal{T}_\nu$, a known  trajectory $\mathcal{T}_\mu$, and its corresponding simplified-and-obfuscated version  $\mathcal{T}_{\widetilde{\mu}}$, we have:
\begin{equation}
cdds(\mathcal{T}_\mu, \mathcal{T}_\nu) \leqslant \overline{cdds}(\mathcal{T}_{\widetilde{\mu}}, \mathcal{T}_\nu) 
\end{equation}
where $\overline{cdds}$ is CDDS with distance threshold $\delta_d + \theta_{sp} + \sqrt{2}\theta_{ob}$.
\end{lemma}
\begin{proof}
The proof is straightforward  and hence is omitted. 
\end{proof}

 \vspace{-2mm}
\subsection{Cost Analysis}
 \vspace{-1mm}
We analyze the cost of Algorithm~\ref{algo:backtrace_join} in terms of the number of node accesses. 
If a query segment spans the whole spatio-temporal space (worst case), all index nodes may be visited, with or without backtracking. However, this rarely happens.  In many cities, points of interest are distributed in a polycentric structure~\cite{polycentric_1, polycentric_2}. Trajectories are likely to gather near local centers, where backtracking helps query the sub-spatial index of the centers. 

The algorithm iterates through the query segments. In each iteration, the costs are spent on finding the common parent node of two adjacent points (and hence adjacent segments) in the query trajectory and on traversing the index to reach leaf nodes.

\begin{figure}[h]
    \centering
    \subfloat[\small{Level 0}\label{fig:backtrace_complexity_level_0}]{
    	\includegraphics[width=0.22\textwidth]{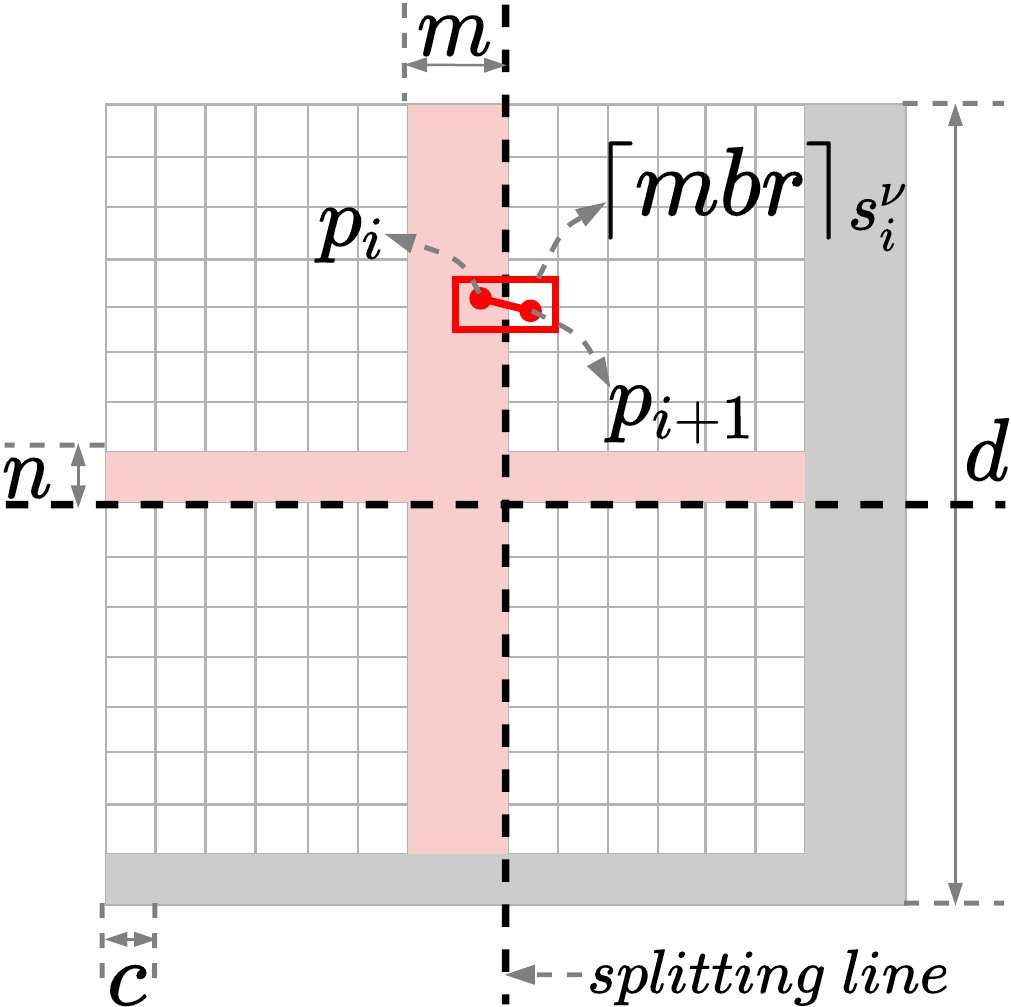}
    }
    \subfloat[\small{Level 1}~\label{fig:backtrace_complexity_level_1}]{
    	\includegraphics[width=0.22\textwidth]{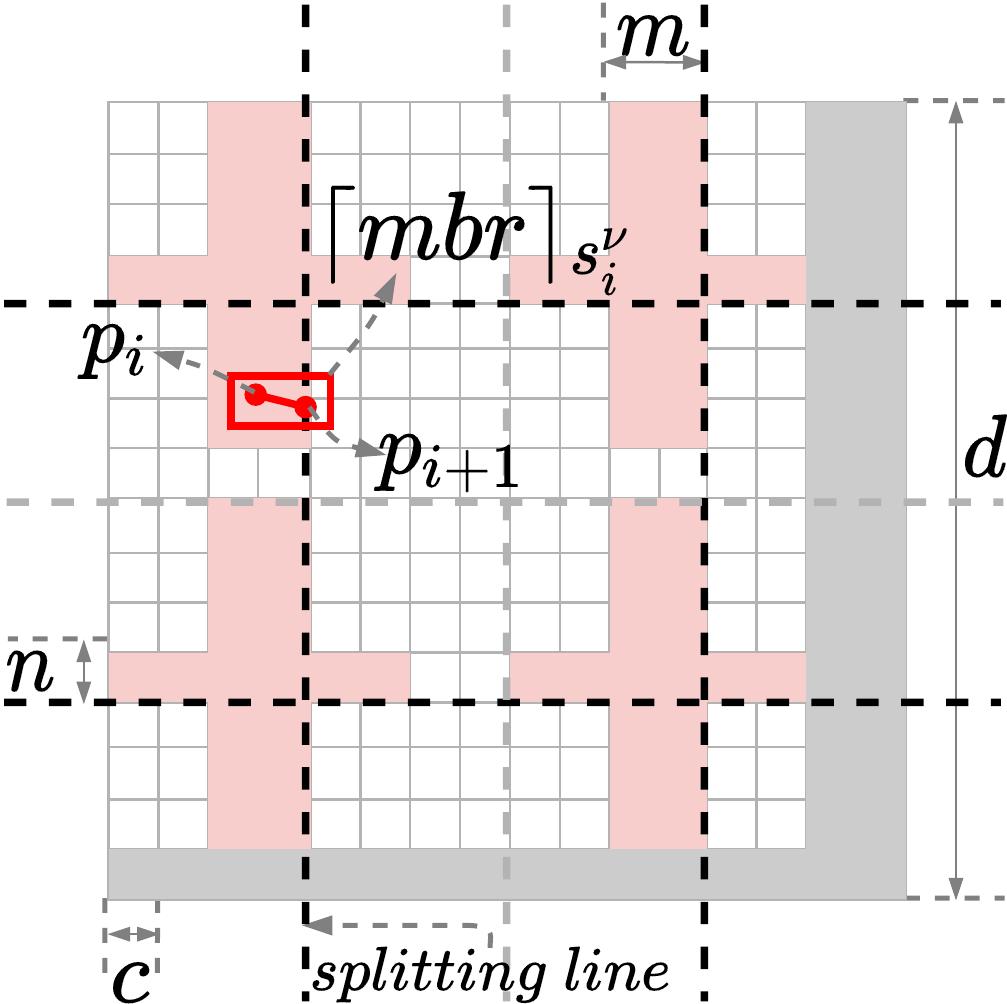}
    }
    \caption{Space division in the  index}\label{fig:backtrace_complexity_level_all}
\end{figure}


First, we derive the backtracking distance $bd$ (i.e., the number of tree levels) between the leaf nodes and the common parent node of two adjacent points. 
We use the expectation $\EX{(bd)}$ to measure the average cost, which is obtained by summing up different backtracking lengths weighted by their probabilities. 

We use Fig.~\ref{fig:backtrace_complexity_level_all} to illustrate how to derive $\EX{(bd)}$ with a space division in the quasi-quadtree. 
The width of the expanded MBR covering the (red) query segment 
is $m$ and the height is $n$, where $m \geqslant n$. 
Let the side length of the space be $d$ and that of a cell be $c$. Here, each cell corresponds to a leaf node. 
In our quasi-quadtree, once the node capacity is exceeded, a node will be divided into four sub-nodes. The space covered by the node is split into four quadrants. We call the vertical and horizontal boundaries (the black dash lines) between the sub-spaces  the \emph{splitting lines}. 

Different splitting lines divide the space at different index levels. The backtracking distance can be derived by the levels of the splitting lines that intersect  with the 
expanded MBR. 
To derive whether an expanded MBR intersects with a splitting line, we  compare the location of the upper left vertex of the MBR with the splitting line. 
In Fig.~\ref{fig:backtrace_complexity_level_0}, when the upper left vertex of the MBR falls in the pink area, the query segment intersects with a splitting line at level 0. The common parent node is the root of the quasi-quadtree, and $bd$ is the height $h$ of this tree, i.e., $\floor{\log_{2}\frac{d}{c}}$.
Fig.~\ref{fig:backtrace_complexity_level_1} shows the level-1 splitting lines, where the backtracking distance is $h - 1$ if the expanded MBR intersects with them. 

Next, we consider the stop condition of segment intersection.
When $m$ and $n$ are large, a query segment may  intersect with splitting lines at multiple lower levels, while common parent nodes cannot be at such levels. 
Let the lowest feasible level of the common parent node be $I$, 
$I = \floor{\log_2\frac{d}{2m}}$. This is because, given a pivot of a query segment in a node, we need at most a common parent node with an MBR side length of $2m$ to also cover the other pivot of the segment. 
In Fig.~\ref{fig:backtrace_complexity_level_1}, if $m > 4c$ and $n > 4c$, the query segment intersects with splitting lines at levels 1 to 4. We only need to consider the case where the query segment intersects with the splitting lines at level 0. 

The probabilities of different backtracking lengths are derived by the ratio of the data space occupied by the pink area of different levels.   
We cut off the grey area in  Fig.~\ref{fig:backtrace_complexity_level_all}, as the query segment is outside the space when its expanded MBR is in this area. 

Then, we compute $\EX(bd)$ by summing up the product of the probability and the backtracking distance at every level:
\vspace{-1mm}
\begin{equation} \label{eq:expected_bd}
\begin{aligned}
    \EX(bd) = \sum\nolimits_{i=0}^{I} \frac{[m(\frac{d}{2^i}-n)+n(\frac{d}{2^i}-m)-mn] \cdot 4^i}{(b-m)(b-n)} \cdot (h - i)
\end{aligned}
\end{equation}
where the fraction part is the probability determined by the size of the pink area, and $h-i$ is the distance between a leaf node and the common parent node. 
We let $\Lambda = \frac{1}{(b-m)(b-n)}$ and then expand Equation~\ref{eq:expected_bd}:
\begin{equation} \label{eq:expected_bd_expand}
\begin{aligned}
    \EX(bd) & = \Lambda \cdot \sum\nolimits_{i=0}^{I} \big\{[m(\frac{d}{2^i}-n)+n(\frac{d}{2^i}-m)-mn] \cdot 4^i \cdot (h - i) \big\} \\
            & = \Lambda \big[ h\sum\nolimits_{i=0}^{I}(2^{i}md+2^{i}nd-3\cdot4^{i}mn) - \sum\nolimits_{i=0}^{I}(2^{i}mdi+2^{i}ndi\\
            & \;\;\; -3\cdot4^{i}mni) \big]
\end{aligned}
\end{equation}
Then, we expand the two terms in square brackets in  Equation~\ref{eq:expected_bd_expand} separately. For the first term, by summing up the geometric series, we can have: 
\vspace{-1mm}
\begin{equation} \label{eq:expected_bd_part1}
\begin{aligned}
        & h\sum\nolimits_{i=0}^{I}(2^{i}md+2^{i}nd-3\cdot4^{i}mn) \\
    ={} & h[md(2^{I+1}-1)+nd(2^{I+1}-1)-mn(4^{I+1}-1)] \\
    ={} & h[md(\frac{d}{m}-1)+nd(\frac{d}{m}-1)-mn(\frac{d^2}{m^2}-1)] \\
    ={} & h(d-m)(d-n)
\end{aligned}
\end{equation}
The second term can be expanded as:
\vspace{-1mm}
\begin{equation} \label{eq:expected_bd_part2}
\begin{aligned}
        & \sum\nolimits_{i=0}^{I}(2^{i}mdi+2^{i}ndi-3\cdot4^{i}mni) \\
    ={} & 2I(d-m)(d-n) - (I+1)\sum\nolimits_{i=1}^{I}(2^{i}md+2^{i}nd-3\cdot4^{i}mn) \\
    ={} & 2I(d-m)(d-n) - (I+1)(d-2m)(d-2n) 
\end{aligned}
\end{equation}
Lastly, we plug Equation~\ref{eq:expected_bd_part1}, \ref{eq:expected_bd_part2} into Equation~\ref{eq:expected_bd_expand}:
\begin{equation} \label{eq:expected_bd_final}
\begin{aligned}
    \EX(bd) & = \Lambda \big[h(d-m)(d-n) - 2I(d-m)(d-n) \\
    & \;\;\; + (I+1)(d-2m)(d-2n) \big] \\
            & = h - 2I + (I+1)\frac{(d-2m)(d-2n)}{(d-m)(d-n)} \\
            & \leqslant h - I + 1
\end{aligned}
\end{equation}
So far, we have that $\EX{(bd)}$ is up to $h-I+1$.

\textbf{Total cost. } To query a trajectory, STS-Join iteratively takes $O((h-I)\bar{|\mathcal{T}|})$ node accesses for backtracking, since only one node is visited at each level, where $\bar{|\mathcal{T}|}$ is the number of segments in a query trajectory.
Besides, it visits $O(4^{h-I}\bar{|\mathcal{T}|})$ nodes while searching for similar segments of a query segment. 
In total, the cost of STS-Join with backtracking is $O(4^{h-I} \bar{|\mathcal{T}|} |\mathcal{D}_q|)$.

\section{Experiments} \label{sec:exp}

\subsection{Experimental Setup}
All experiments are conducted on a 64-bit machine running Ubuntu 20.04 LTS with a 6-core AMD Ryzen 5 CPU, 32 GB RAM, and a 500GB SSD. All algorithms are implemented in C++ and run in main memory.

\textbf{Datasets.} To the best of our knowledge, there is no public contact-tracing datasets. We use two real trajectory datasets instead: \textbf{DiDi}~\cite{url_didi_dataset}
and \textbf{GeoLife}~\cite{url_geolife_dataset}.
DiDi contains vehicle trajectories in Chengdu, a capital city in China. 
We randomly sample trajectories recorded in the first week of November, 2016 to generate a dataset with 500k trajectories and 67.7 million segments (this is limited by our memory size). The trajectories come from 201k users (each is considered a client in the experiments), i.e., there are 2.5 trajectories per client on average. 
There are 135 segments per trajectory on average. The average length and time span of a segment are 30.0 meters and 4.5 seconds, respectively. 

GeoLife contains trajectories of different transport modes (e.g., walking, driving, and cycling) recorded mainly in Beijing, China. We only keep the trajectories within the Fifth Ring Road of Beijing. There are some 11k trajectories in this area and very few outside. We randomly sample 10k trajectories from them to form the GeoLife dataset used in the experiments. 
We shift the trajectory times into the same week without changing their time span to increase the density in time. The trajectories come from 182 users (clients), i.e., there are 55 trajectories per client on average.
There are 9.0 million segments in these trajectories and 896 segments per trajectory on average. The average length and time span of a segment are 12.2 meters and 9.2 seconds, respectively.  
The two datasets are summarized in Table~\ref{tab:datasets}. 
We further generate their random subsets for scalability tests. See Table~\ref{tab:exp_parameter} for the sizes of the subsets.


\setlength\tabcolsep{2pt} 
\begin{table}[htp]
\centering
\caption{Datasets} \label{tab:datasets}
\begin{tabular}{r|c|c|c|c}
\hline
\multirow{2}{*}{\textbf{Dataset}} & \multicolumn{4}{c}{\textbf{Configuration}}\\\cline{2-5}
& \# traj. & \# seg. & Avg. seg. length & Avg. seg. time span \\
\hline
\hline
DiDi  &  500k & 67,654k & 30.0 meters & 4.5 seconds\\ \cline{2-5} 
GeoLife  & 10k & 8,981k & 12.2 meters & 9.2 seconds
   \\ \hline
\end{tabular}
\end{table}

\textbf{Algorithms.}
There is no existing algorithm for STS-Join. We adapt two techniques which result in three baseline algorithms. Meanwhile, we replace the join predicates in these algorithms with ours, so \emph{all algorithms return the same accurate query results}. We test three variants of our proposed STS-Join algorithm to confirm the effectiveness of the proposed backtracking-based query algorithm and the CDDS-based pruning strategy.
\begin{itemize}
    \item \textbf{3DR}:
    The 3DR-tree~\cite{SpatialTemporalSTRree} extends the R-tree by adding time ranges as an extra dimension. We use it to index segments in $\mathcal{D}_p$ and run an R-tree join like algorithm for queries. 
    \item \textbf{ALSTJ}: As described in Section~\ref{sec:relatedwork}, 
    ALSTJ~\cite{join_subtraj_edbt2020} (an R-tree variant) is the closest study to ours, which also considers sub-trajectory join. We replace the join predicate in its query algorithm with ours to process STS-Join. 
    \item \textbf{ALSTJ-T}: The original ALSTJ index does not consider the time dimension. We improve it by adding a B-tree-like temporal index at its top levels like our STS-Index. 
    \item \textbf{STS}: STS denotes the naive STS-Join without backtracking or CDDS-based pruning as described at the start of Section~\ref{sec:traj_join_algo}.
    \item \textbf{STS-BT}: STS-BT is the backtracking-based STS-Join (no  pruning), described in Algorithm~\ref{algo:backtrace_join}.
    \item \textbf{STS-BTB}: STS-BTB further prunes matched trajectory segment pairs from being sent to the clients by a CDDS upper bound as described in Section~\ref{subsec:sim_bound}.
\end{itemize}
We set the node capacity of the baselines and the B-trees in STS-Index to 100.  We set the B-tree and the bitmap time interval lengths in STS-Index to 1,800 and 30 seconds, respectively.
The client-server model of STS-Join is  simulated, since it aims to enhance privacy but not computation capacity.

\textbf{Parameter settings.} 
Table~\ref{tab:exp_parameter} summarizes parameters tested in the experiments, where the default values are in bold.


We follow the closest work~\cite{join_subtraj_edbt2020} to compare the algorithm response times, and we also measure the I/O cost that is a widely used indicator in spatial indexing studies~\cite{btree, rrtre_vldb,rrtre_tods} if using external memory based implementation.
Besides, we study the communication cost by recording the number of segments sent to the clients for final verification and the number of clients receiving such segments.
\emph{All algorithms return the same result set under the same query setting, which verifies the correctness of our methods.}

\setlength\tabcolsep{4pt} 
\begin{table}[htp]
\centering
\caption{{Parameters and Their Values}}\label{tab:exp_parameter}
\begin{tabular}{l|l}
\hline
 \textbf{Parameter} & \textbf{Settings}  \\ \hline \hline
$|\mathcal{D}_p|$ (DiDi) & 100k, \textbf{200k}, 300k, 400k, 500k \\ \hline
$|\mathcal{D}_p|$ (GeoLife) & 2k, \textbf{4k}, 6k, 8k, 10k \\ \hline
$|\mathcal{D}_q|$ (DiDi) & 100, \textbf{200}, 300, 400, 500 \\ \hline
$|\mathcal{D}_q|$ (GeoLife) & 2, \textbf{4}, 6, 8, 10 \\ \hline
$\delta_d$ & 10, \textbf{20}, 30, 40, 50  (meters)\\ \hline
$\delta_t$ & 0, 10, \textbf{20}, 30, 40, 50 (seconds) \\ \hline
$\theta_{sp}$ & 0, 10, \textbf{20}, 30, 40, 50 (meters) \\ \hline
$\theta_{ob}$ & 0, 10, \textbf{20}, 30, 40, 50 (meters) \\ \hline
\end{tabular}
\end{table}
\vspace{-3mm}

\subsection{Join Performance} 
We set the default sizes of the known DiDi and GeoLife datasets $\mathcal{D}_p$ to 200k and 4k, respectively.  
The query datasets $\mathcal{D}_q$ are randomly sampled from the original DiDi and GeoLife datasets with the same time ranges as the known datasets $\mathcal{D}_p$. Their sizes are set to be proportional to those of the respective known DiDi and GeoLife datasets. 
The default sizes of DiDi and GeoLife query datasets are 
 200 and 4 (i.e., $0.1\%$ of $|\mathcal{D}_p|$), respectively. The GeoLife query datasets might seem small, but there are some 900 segments per trajectory on average, which are sufficient to show the performance difference of the algorithms. 

\begin{figure}[h]
    \vspace{-3mm}
    \centering
    \subfloat[Response time (DiDi)~\label{fig:exp2_time}]{
    	\includegraphics[width=0.23\textwidth]{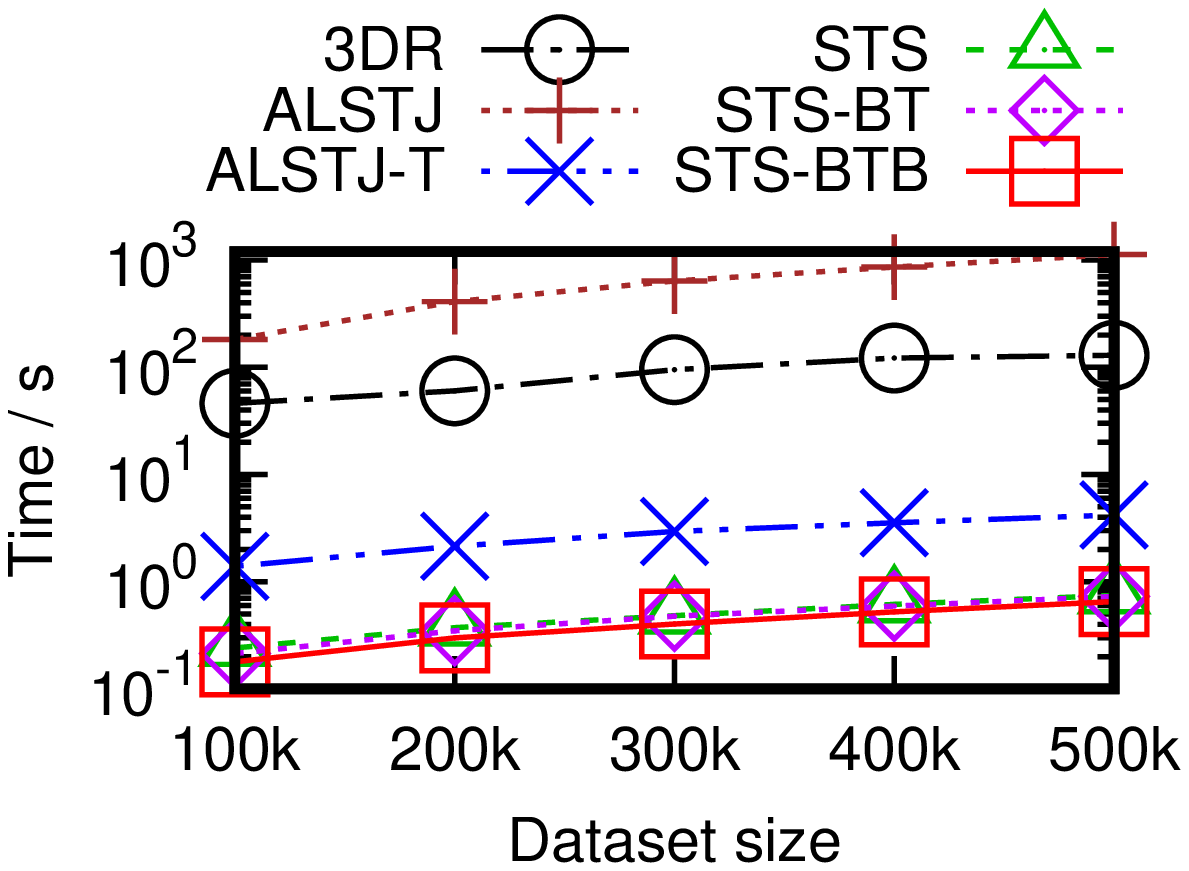}
    }
    \subfloat[{Number of node accesses (DiDi)}~\label{fig:exp2_io}]{
    	\includegraphics[width=0.23\textwidth]{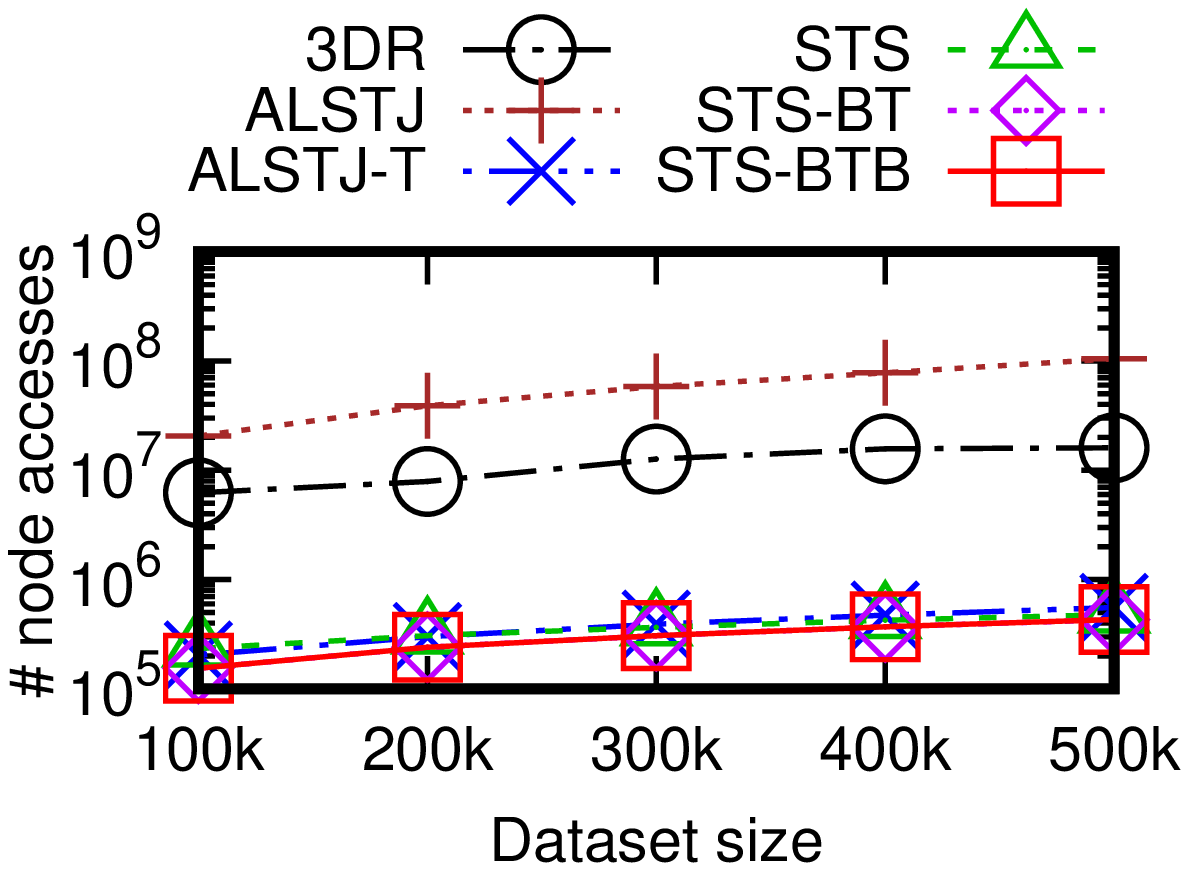}
    }\\
    \vspace{-3mm}
    \subfloat[{Response time (GeoLife)}~\label{fig:exp2_geolife_time}]{
    	\includegraphics[width=0.23\textwidth]{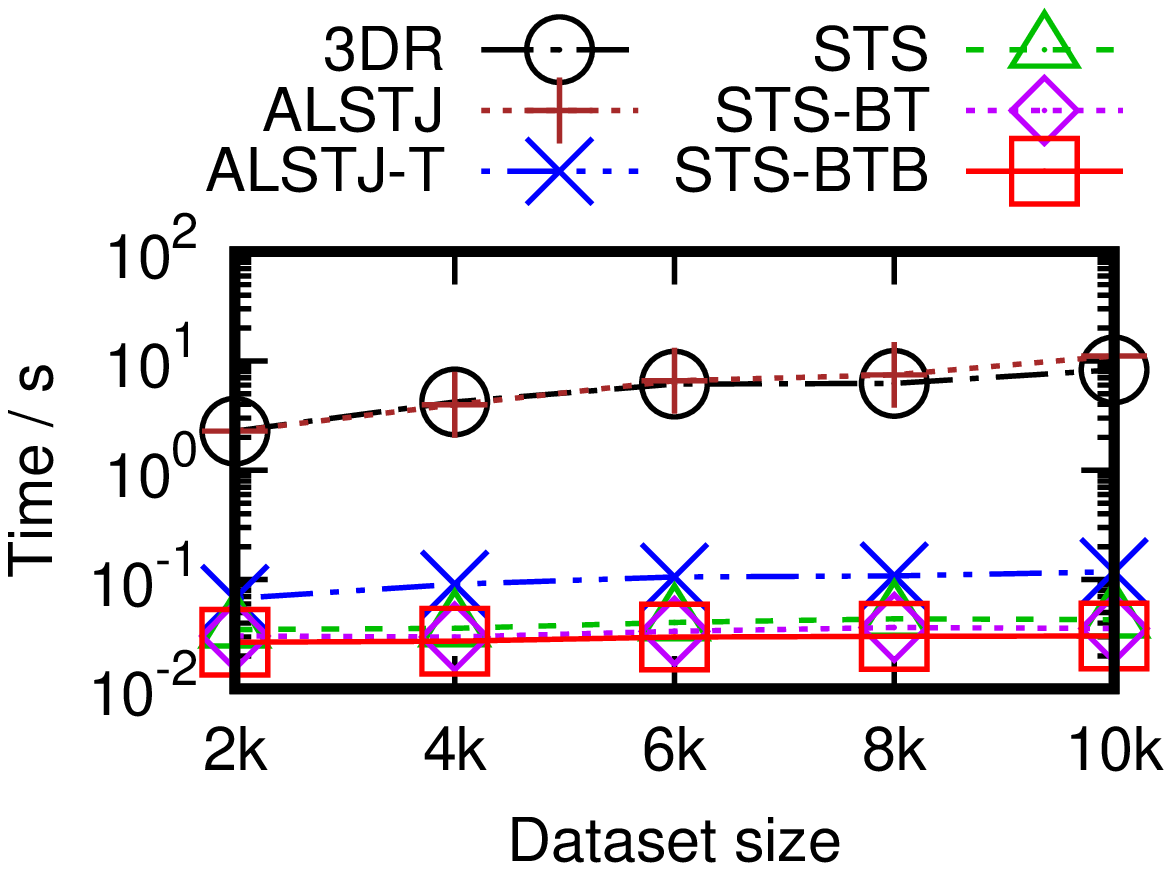}
    }
    \subfloat[{Number of node accesses (GeoLife)}~\label{fig:exp2_geolife_io}]{
    	\includegraphics[width=0.23\textwidth]{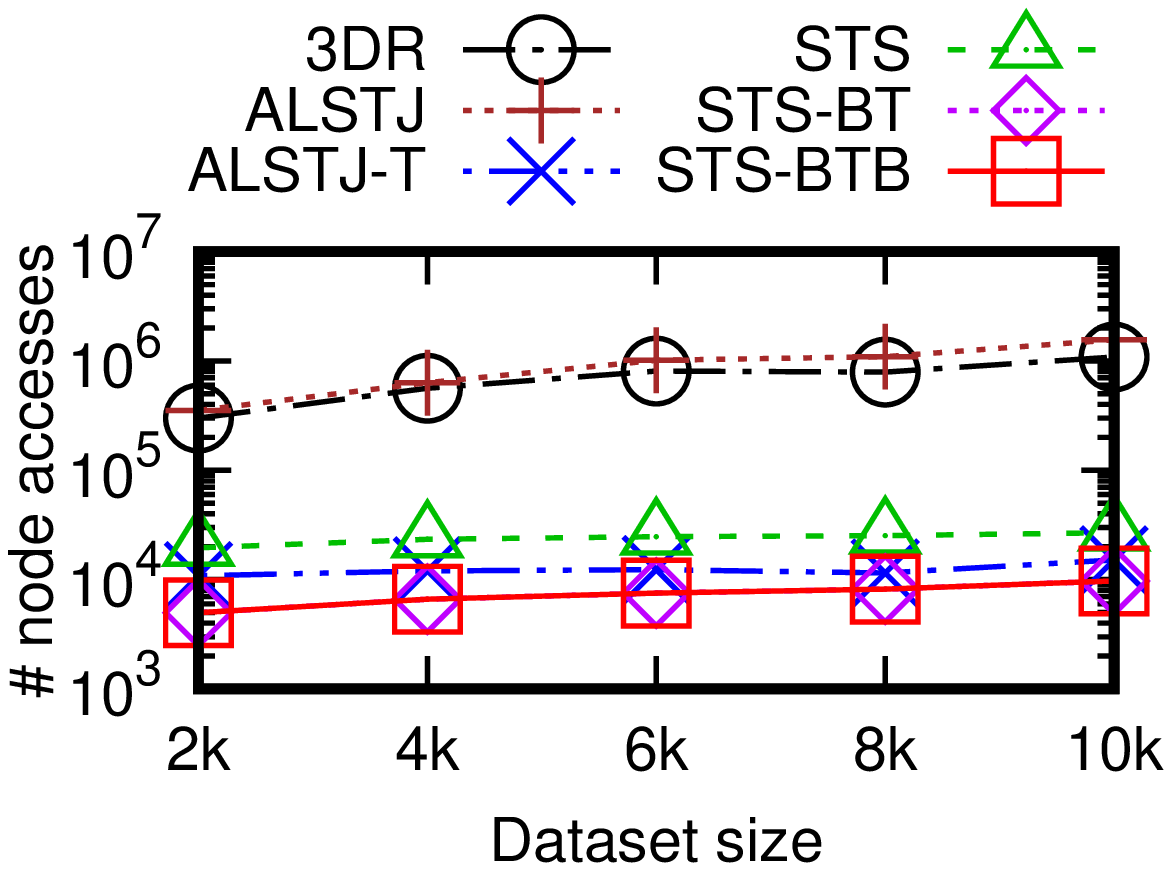}
    }
    \caption{Query costs vs. known dataset size}\label{fig:join_data_size}
\end{figure}

\textbf{Varying the known dataset size $|\mathcal{D}_p|$. }
Fig.~\ref{fig:join_data_size} shows the query costs, which increase with $|\mathcal{D}_p|$. Our STS-BTB algorithm with both backtracking and CDDS-based pruning outperforms all competitors consistently in response time. It also takes the fewest node accesses among all algorithms except our STS-BT, which has the same node accesses as STS-BTB since the CDDS-based  pruning does not impact node accesses. 
The advantage is up to three orders of magnitude. 
Comparing STS-BTB with STS-BT and STS, it achieves 
$\sim$14\% and $\sim$18\% lower running times,  which confirms the effectiveness of our CDDS-based pruning to reduce the computational costs.
STS has more node accesses ($\sim$24\% on average) than STS-BT, confirming the effectiveness of our backtracking strategy to reduce node accesses. 
STS has very similar running times to those of STS-BT, as the algorithms are running in memory, where the time saved on node accesses is lost in the more complex backtracking procedure of STS-BT. 
STS also outperforms the baselines in running time, while it has slightly more node accesses than ALSTJ-T when $|\mathcal{D}_p| \leqslant$ 200k on DiDi datasets. This is because ALSTJ-T has a much larger node capacity in its index, i.e., 100, while we use 4 in our quasi-quadtrees. The response time of ALSTJ-T is still up to 6.7 times larger than those of STS and STS-BT. This is because of its worse pruning power -- its index allows overlapping node MBRs while ours does not. ALSTJ is even worse, i.e., up to three orders of magnitude slower and two orders of magnitude more node accesses than STS-BTB, because of its lack of pruning power in the time dimension. 3DR considers the time dimension but is not optimized for trajectories.
It runs up to 240 times slower and has up to 43 times more node accesses than STS-BTB. 

\begin{figure}[h]
    \vspace{-6mm}
    \centering
        \subfloat[{Response time (DiDi)}~\label{fig:exp3_time}]{
    	\includegraphics[width=0.23\textwidth]{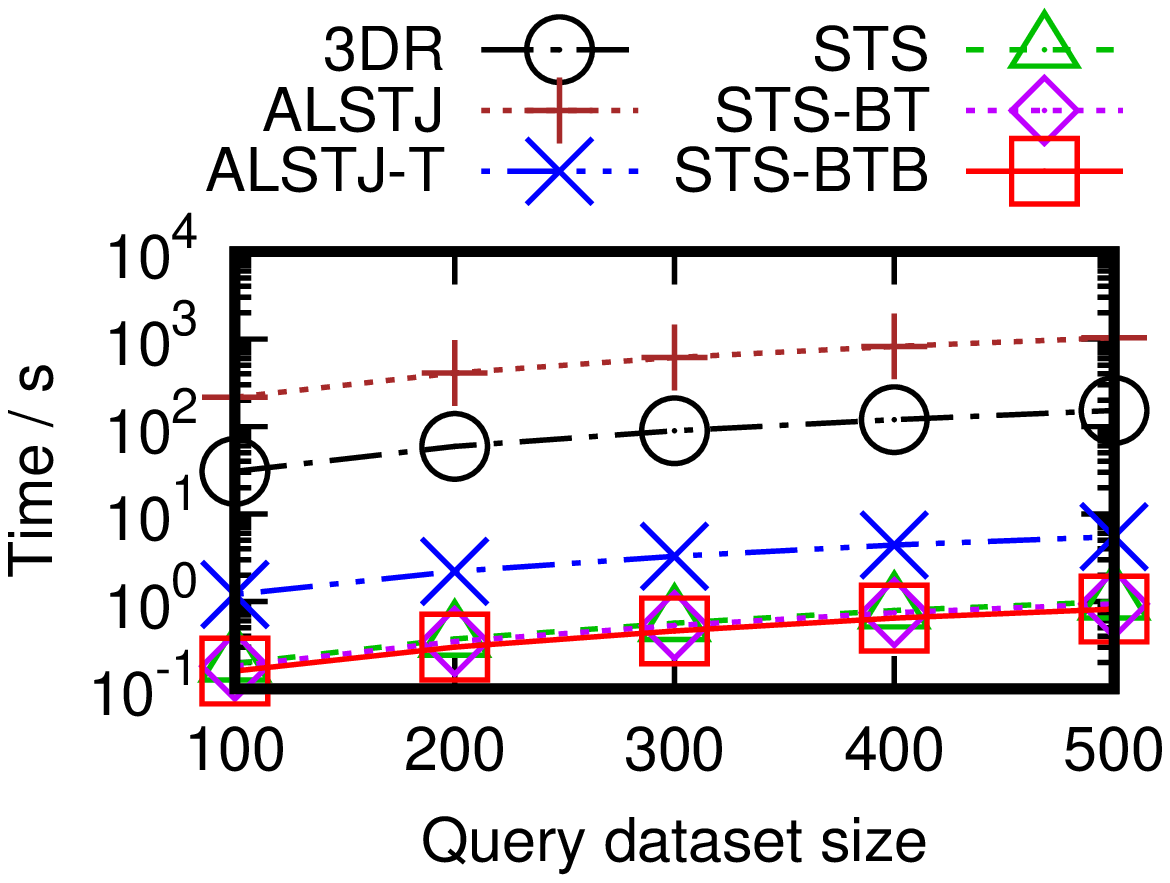}
    }
    \subfloat[{Number of node accesses (DiDi)}~\label{fig:exp3_io}]{
    	\includegraphics[width=0.23\textwidth]{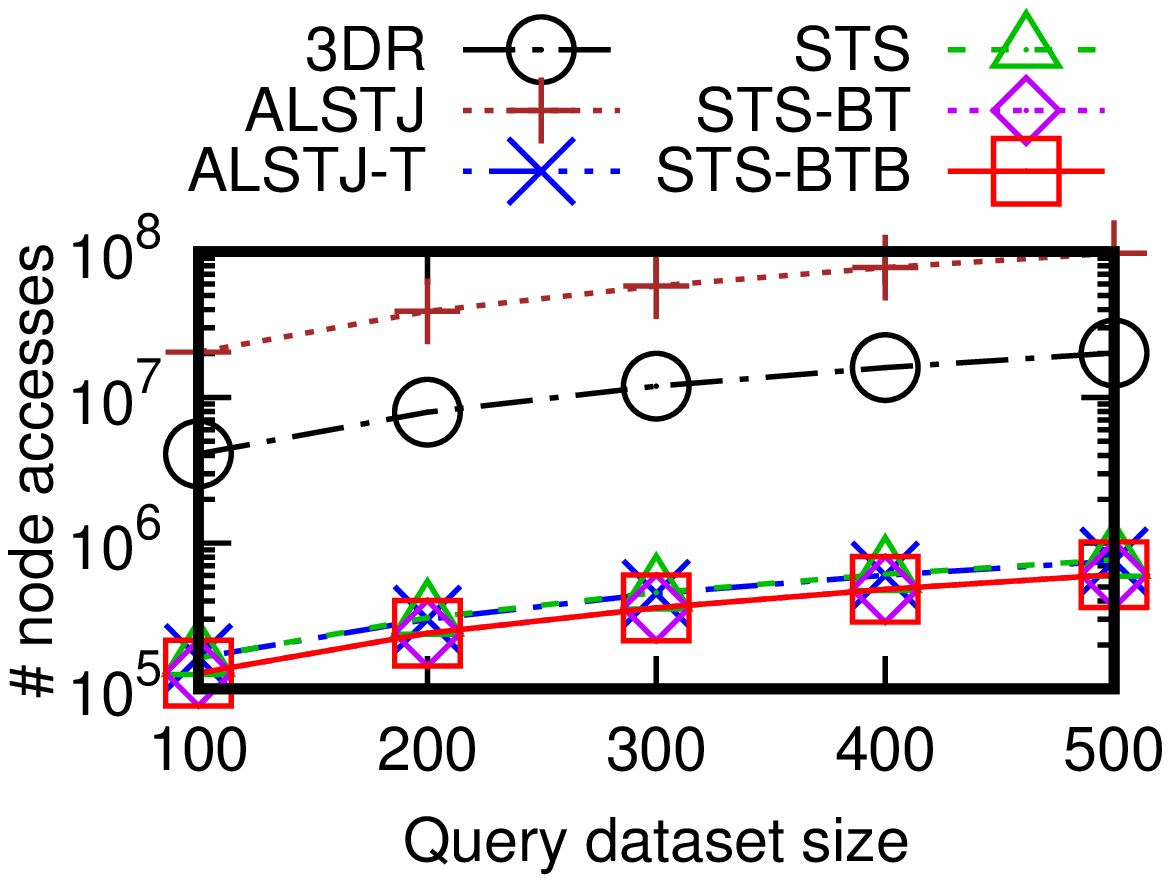}
    }\\
    \vspace{-2mm}
    \subfloat[{Response time (GeoLife)}~\label{fig:exp3_geolife_time}]{
    	\includegraphics[width=0.23\textwidth]{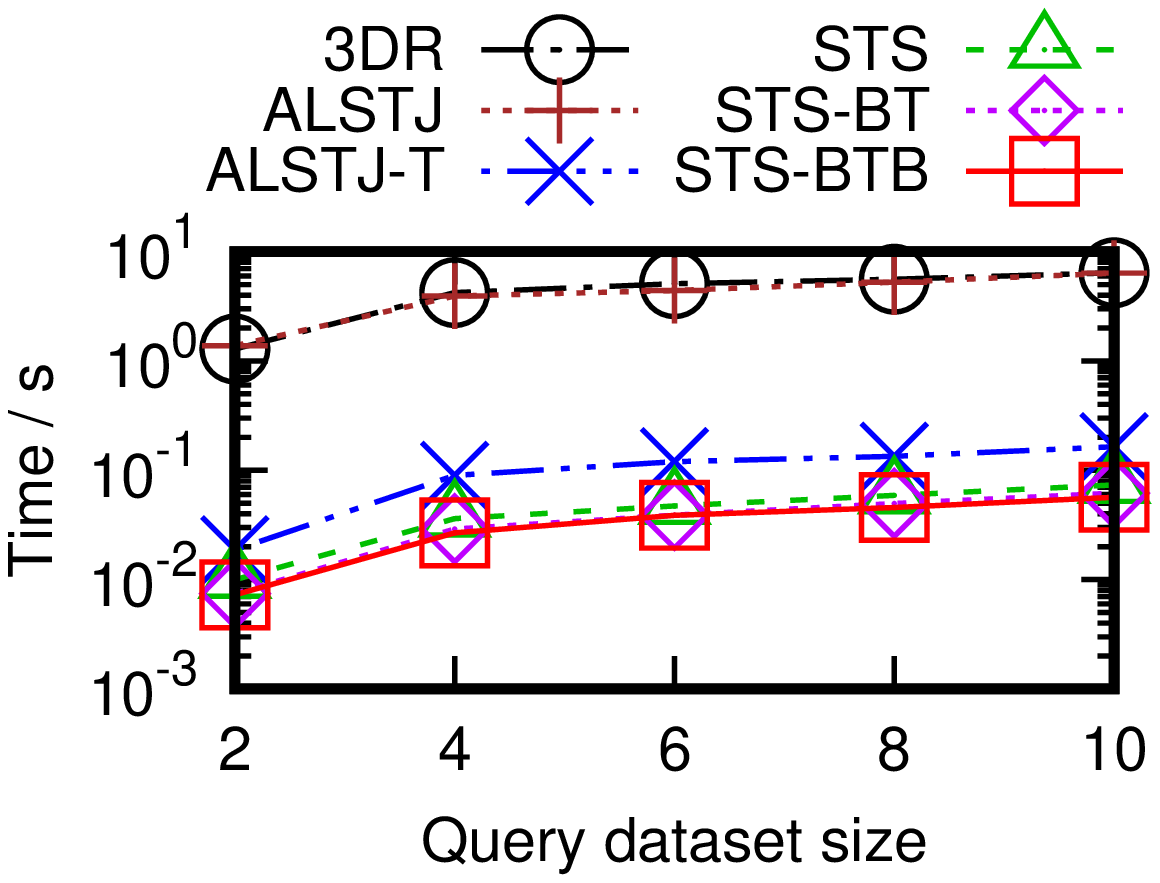}
    }
    \subfloat[{Number of node accesses (GeoLife)}~\label{fig:exp3_geolife_io}]{
    	\includegraphics[width=0.23\textwidth]{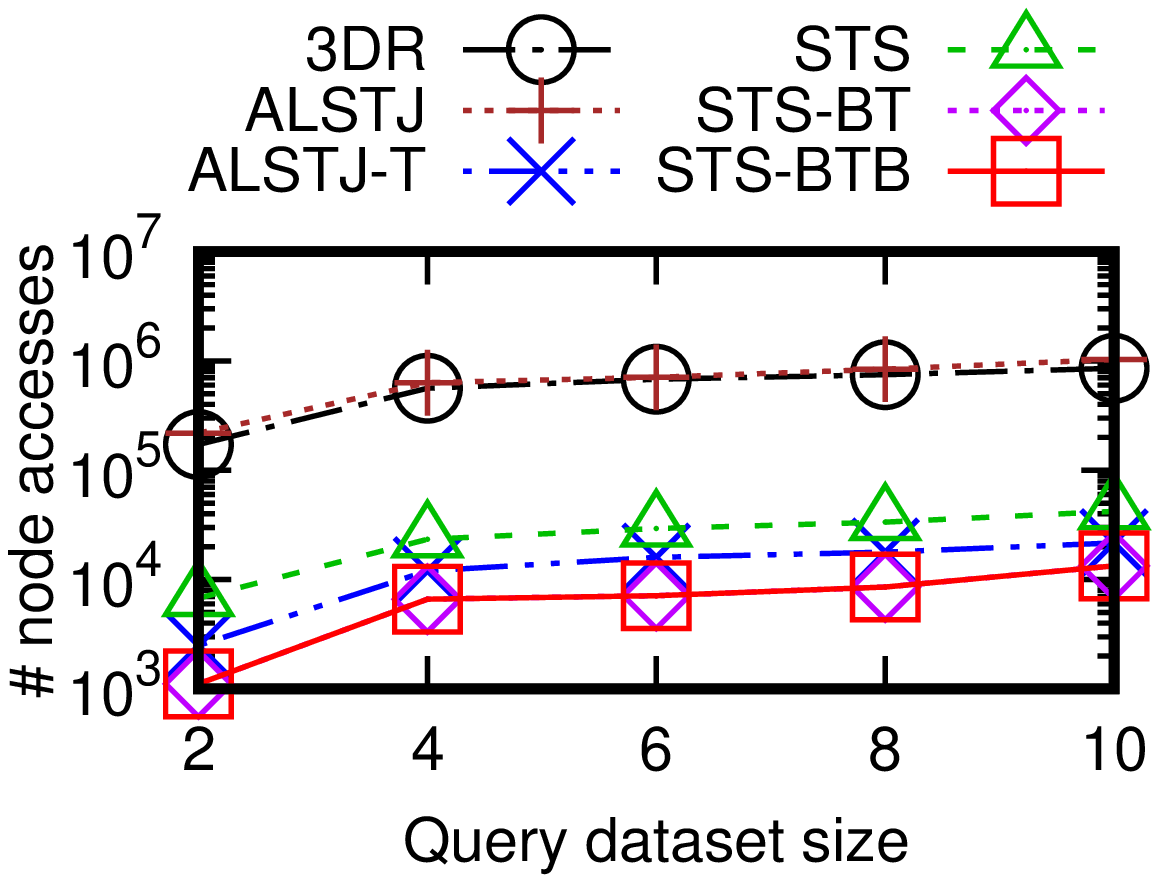}
    }
    \caption{Query costs vs. query dataset size}\label{fig:join_query_size}
\end{figure}

\textbf{Varying the query dataset size $|\mathcal{D}_q|$. }
Next, we vary the query dataset size $|\mathcal{D}_q|$. Fig.~\ref{fig:join_query_size} shows that the query costs increase with $|\mathcal{D}_q|$, which is expected. 
The relative performance among the algorithms is very similar to that when varying $|\mathcal{D}_p|$.
STS, STS-BT, and STS-BTB outperform all competitors on both datasets in response time (up to three orders of magnitude), while ALSTJ-T has a slightly smaller number of node accesses than that of STS. 
These confirm the superiority of the proposed algorithms and the effectiveness of the proposed backtracking and CDDS-based pruning techniques to reduce the algorithm costs. 

Since the query performance patterns are similar on both DiDi and GeoLife data, we omit the figures for GeoLife in the following query experiments for conciseness.

\begin{figure}[h]
    \vspace{-3mm}
    \centering
    \subfloat[{Response time}~\label{fig:exp4_1_time}]{
    	\includegraphics[width=0.23\textwidth]{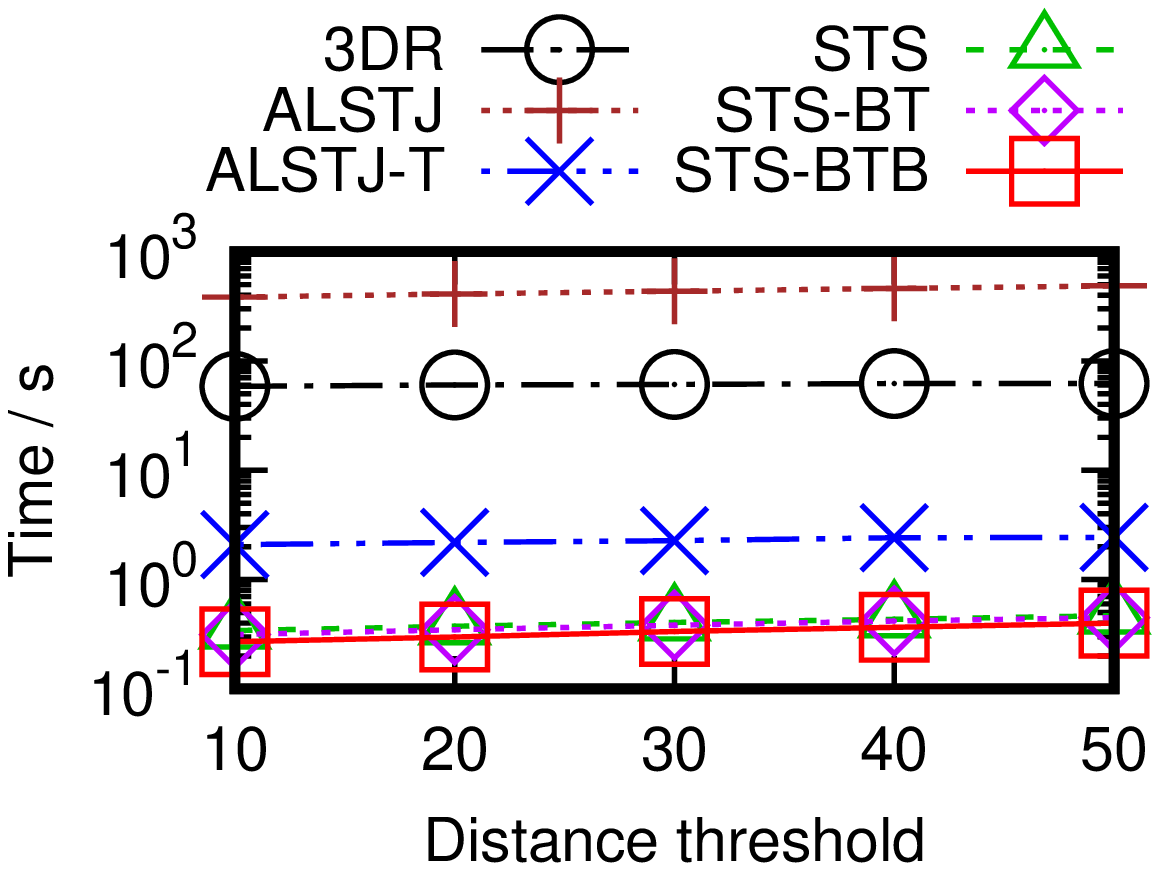}
    }
    \subfloat[{Number of node accesses}~\label{fig:exp4_1_io}]{
    	\includegraphics[width=0.23\textwidth]{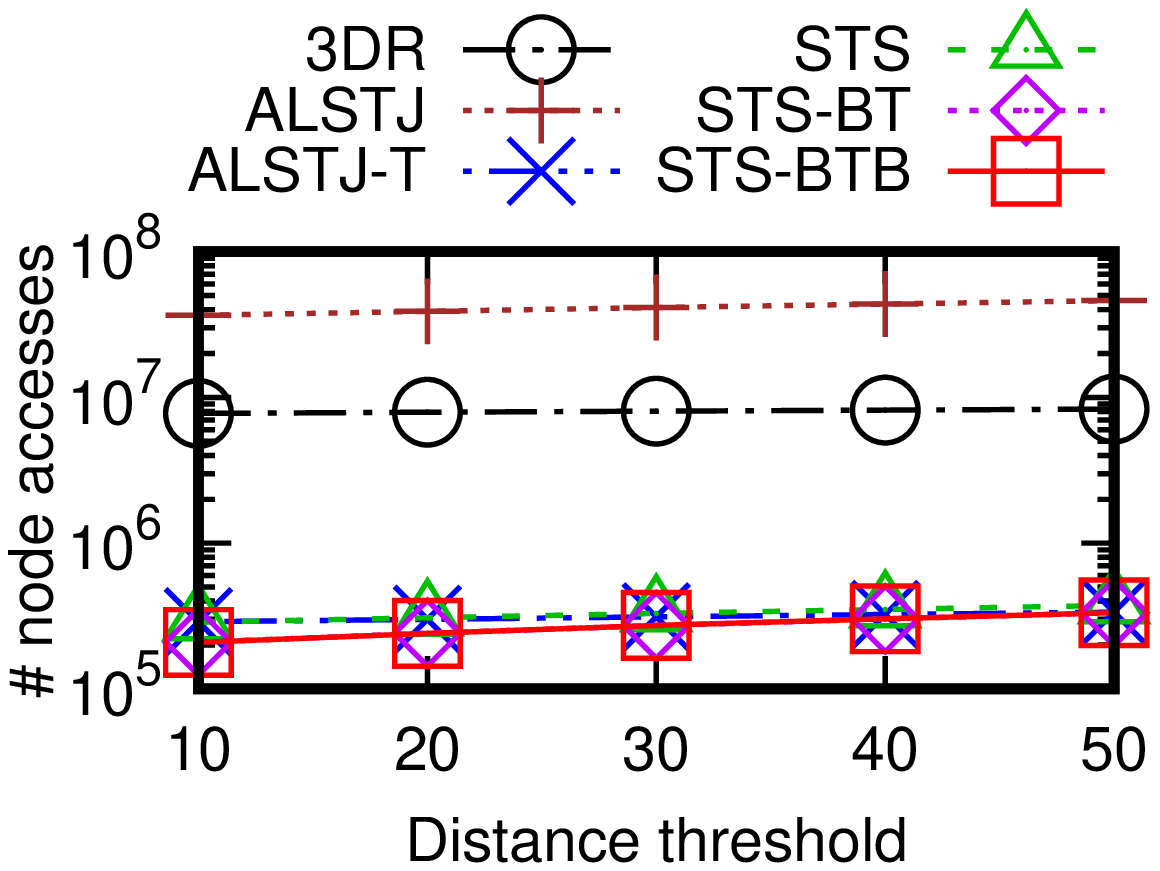}
    }\\
    \vspace{-2mm}
    \subfloat[Number of segments sent~\label{fig:exp8_1_segment}]{
    	\includegraphics[width=0.23\textwidth]{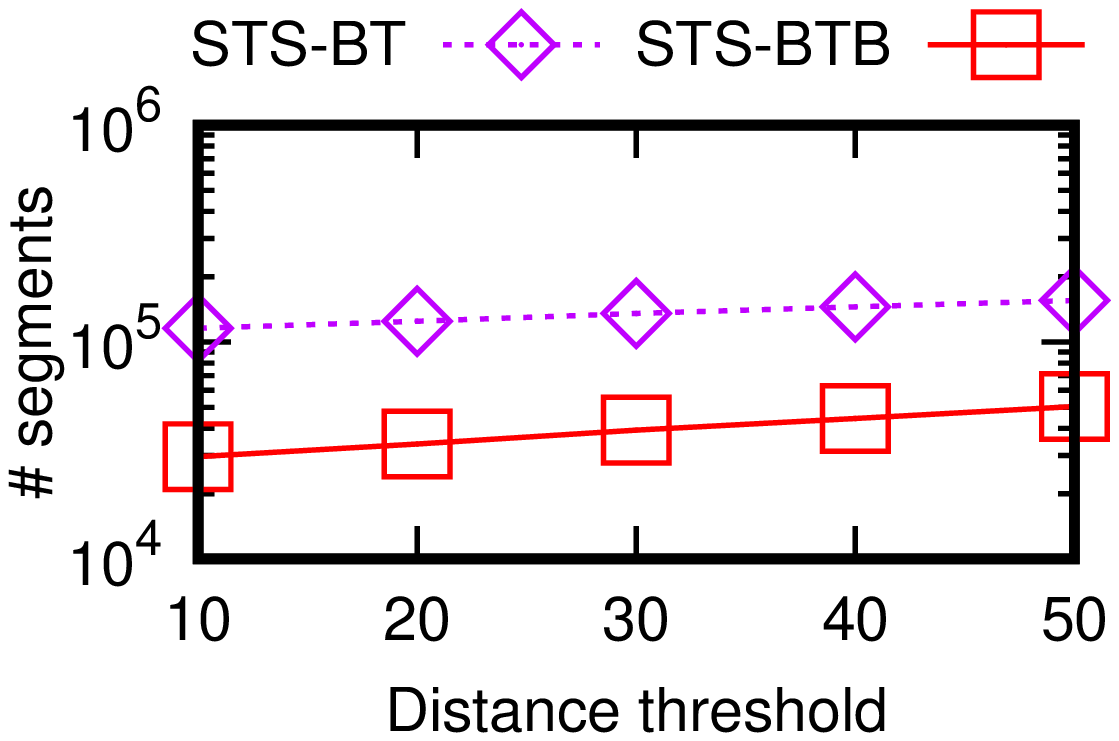}
    }
    \subfloat[Number of receiving clients~\label{fig:exp8_1_client}]{
    	\includegraphics[width=0.23\textwidth]{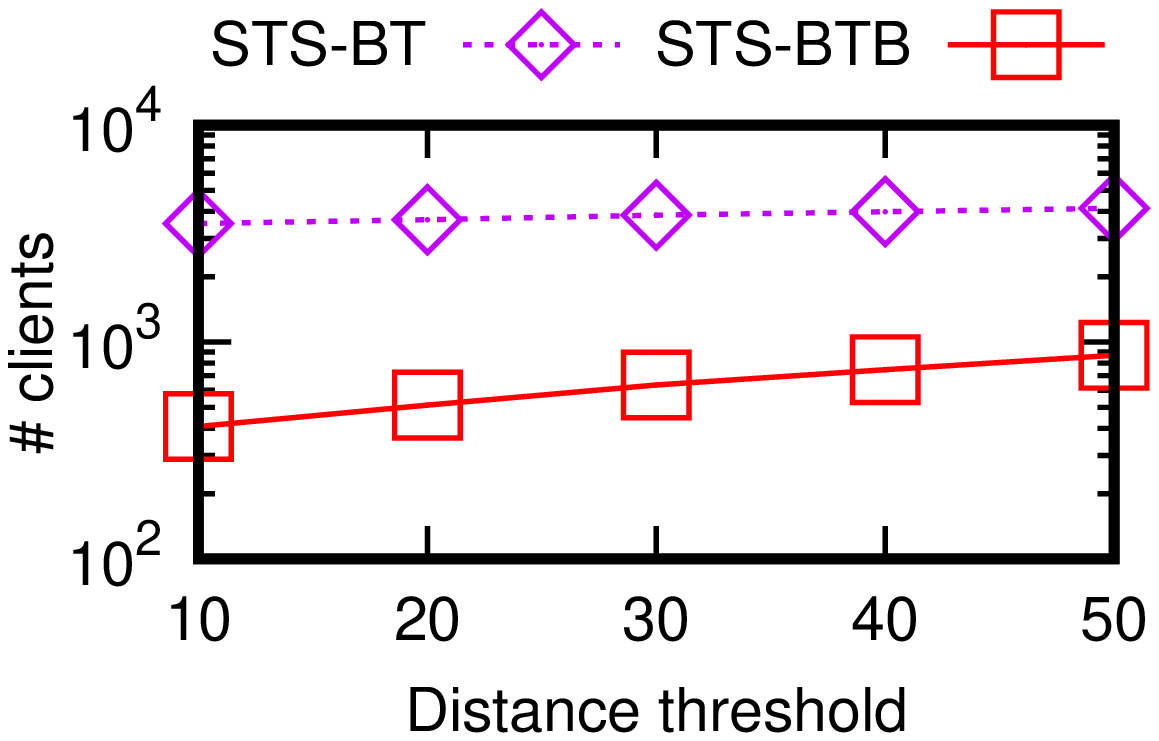}
    }
    \caption{Query costs vs. query distance threshold}\label{fig:join_distance_threshold}
\end{figure}

\textbf{Varying the query distance threshold $\delta_d$.}  Fig.~\ref{fig:join_distance_threshold} shows the algorithm performance as $\delta_d$ increases from 10 to 50 meters. The query costs of all algorithms increase with $\delta_d$, because more trajectory pairs satisfy a larger $\delta_d$ and need to be returned. 
The relative performance between the proposed STS algorithms and the competitors are similar to those observed above. Like before, STS-BTB is up to 22\% and 14\% faster than STS and STS-BT, respectively, while STS-BT and STS-BTB both have the fewest node accesses. We also study the impact of $\delta_d$ on the communication cost, which is reflected by the number of segments that are sent to the clients and the number of clients receiving at least one segment, as reported in 
Figs.~\ref{fig:exp8_1_segment} and~\ref{fig:exp8_1_client}, respectively.  
We see that, while the communication costs of STS-BTB  increase with $\delta_d$ which is expected, they are at least 3.1 and 4.8 times lower than those of STS-BT (for $\delta_d \le 50$) in terms of the number of segments sent and the number of clients receiving the segments, respectively. Here, STS has been omitted since it has the same communication costs as STS-BT, and the same applies to the figures below. 
The results further confirm the effectiveness of our CDDS-based pruning.

\begin{figure}[h]
    \vspace{-3mm}
    \centering
    \subfloat[{Response time}~\label{fig:exp4_5_time}]{
    	\includegraphics[width=0.23\textwidth]{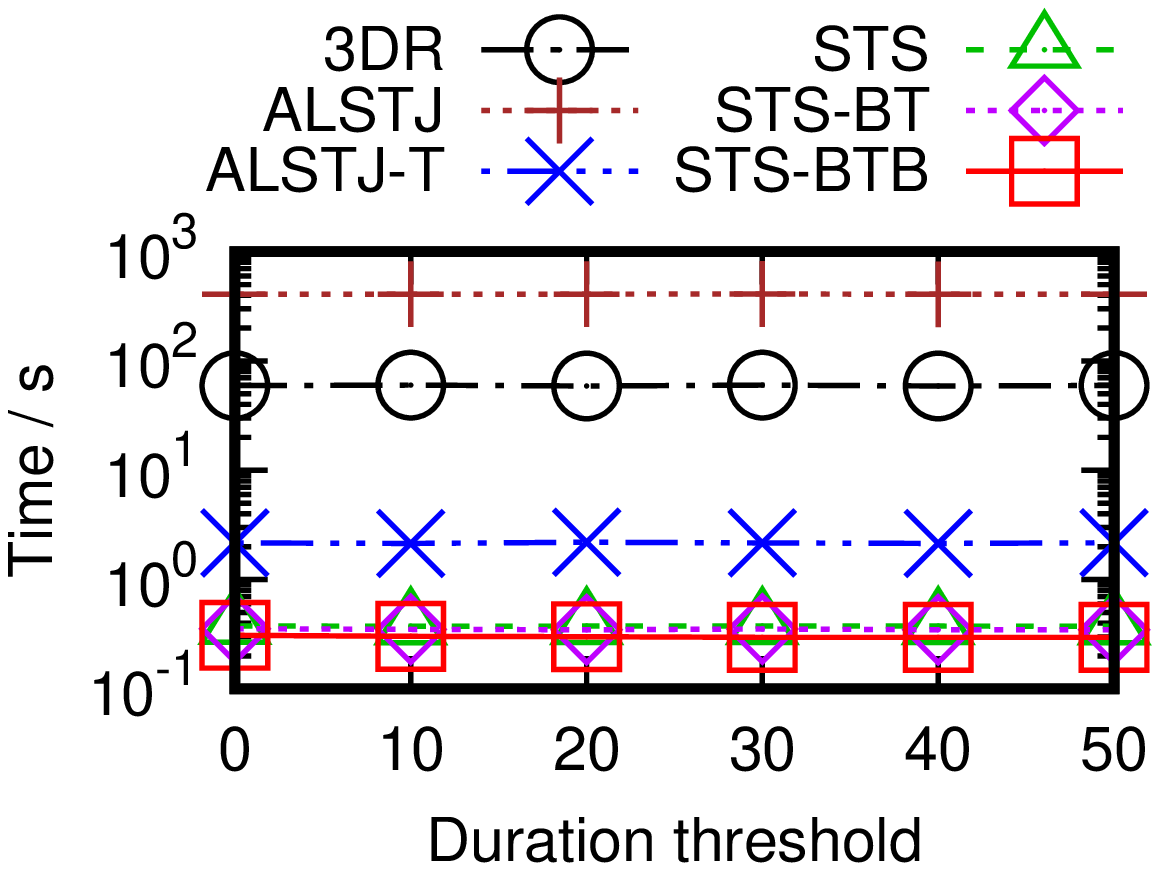}
    }
    \subfloat[{Number of node accesses}~\label{fig:exp4_5_io}]{
    	\includegraphics[width=0.23\textwidth]{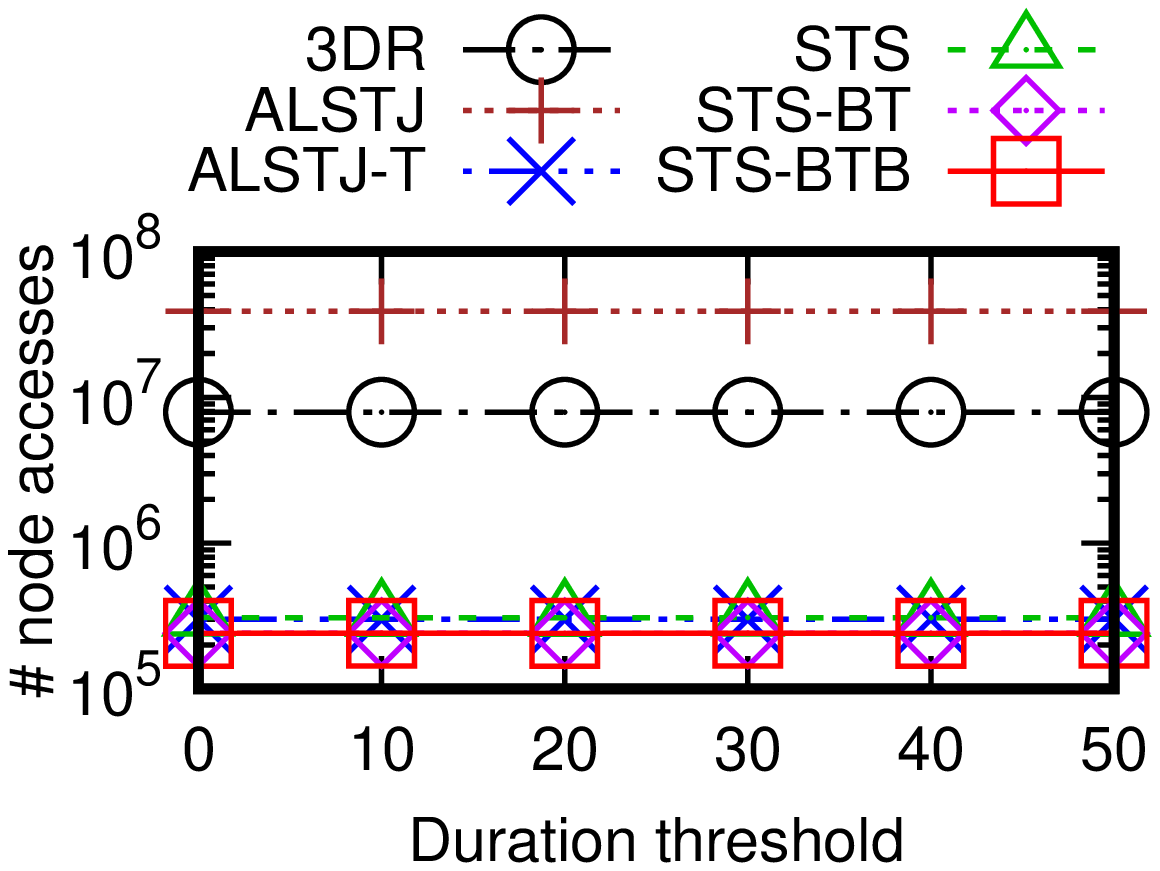}
    }\\
    \vspace{-2mm}
    \subfloat[{Number of segments sent}~\label{fig:exp8_5_segment}]{
    	\includegraphics[width=0.23\textwidth]{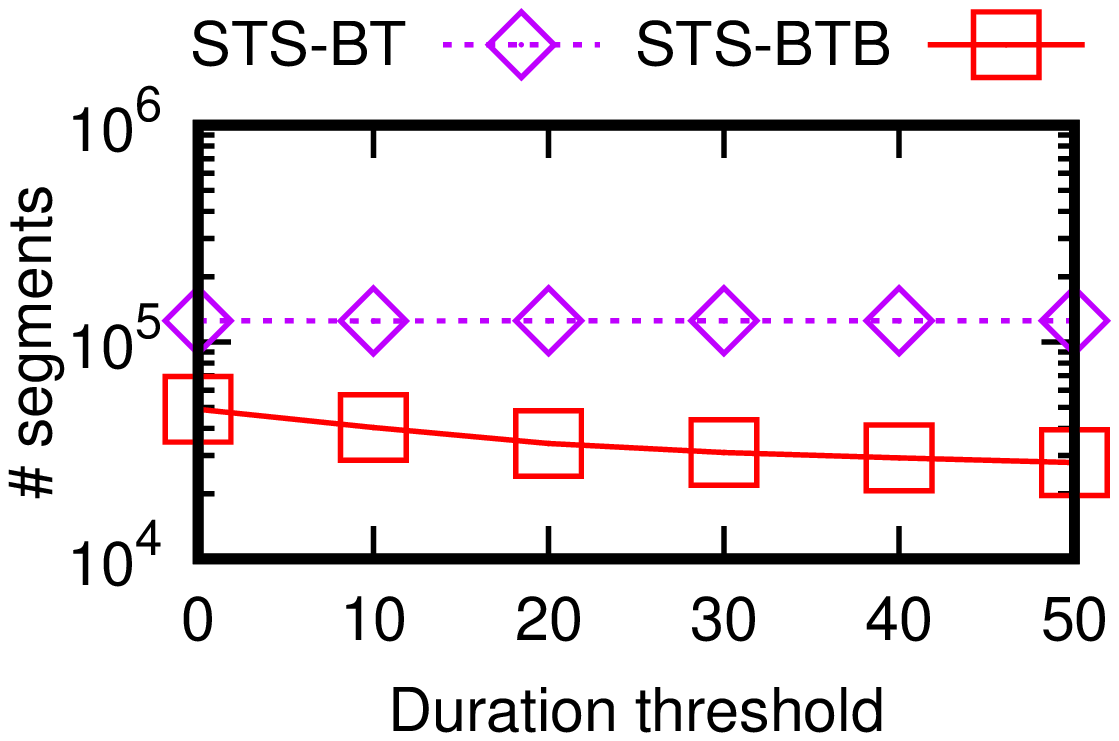}
    }
    \subfloat[Number of receiving clients~\label{fig:exp8_5_client}]{
    	\includegraphics[width=0.23\textwidth]{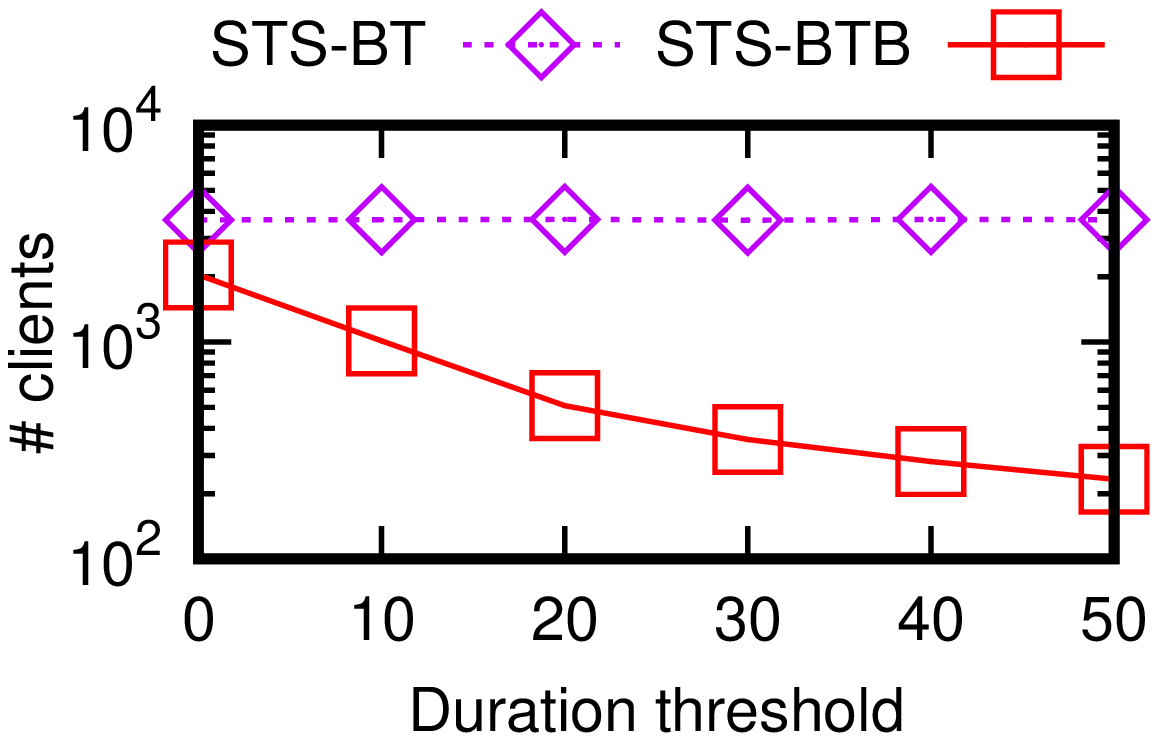}
    }
    \caption{Query costs vs. query duration threshold}\label{fig:join_duration_threshold}
\end{figure}

\textbf{Varying the query duration threshold $\delta_t$. }
Figs.~\ref{fig:exp4_5_time} and~\ref{fig:exp4_5_io} show that the join query costs are relatively stable as the query duration threshold $\delta_t$ increases from 0 to 50 seconds. This is because all algorithms compare the CDDS of trajectories with $\delta_t$ at almost their final stages. The value of $\delta_t$ has little impact on the time and node access costs. The relative performance among the algorithms is again very similar to that in the previous figures. 
From Figs.~\ref{fig:exp8_5_segment} and~\ref{fig:exp8_5_client}, we see that the communication costs of STS-BT is relatively stable, while those of STS-BTB decrease quickly as $\delta_t$ grows. This is because a larger $\delta_t$ makes the query predicate more difficult to be satisfied, even for the  CDDS used in pruning which has been inflated to guarantee no false negatives. This indicates that our CDDS-based pruning strategy may gain more pruning power in application scenarios with larger $\delta_t$ values.

\textbf{Varying the simplification threshold $\theta_{sp}$.} 
We also study the impact of the simplification threshold on STS-Join. As shown in Figs.~\ref{fig:exp4_2_time} and~\ref{fig:exp4_2_io}, when $\theta_{sp}$ increases from 0 to 50 meters, the query costs first drop and then increase for all STS algorithm variants. The drop is because introducing simplification helps reduce the number of segments in STS-Index and hence the query costs. However, as $\theta_{sp}$ grows, the data segments become longer which leads to larger node MBRs (and greater MBR expansion to ensure no false negatives) and reduced the index pruning power. This explains for the increase in query costs.
STS-BTB and STS-BT have the same node accesses since 
the CDDS-based pruning has no impact on node accesses. Meanwhile, STS-BTB again outperforms STS-BT in the communication costs, thus showing the robustness of STS-BTB (cf. Figs.~\ref{fig:exp8_2_segment} and~\ref{fig:exp8_2_client}).

\begin{figure}[h]
    \vspace{-6mm}
    \centering
    \subfloat[{Response time}~\label{fig:exp4_2_time}]{
    	\includegraphics[width=0.23\textwidth]{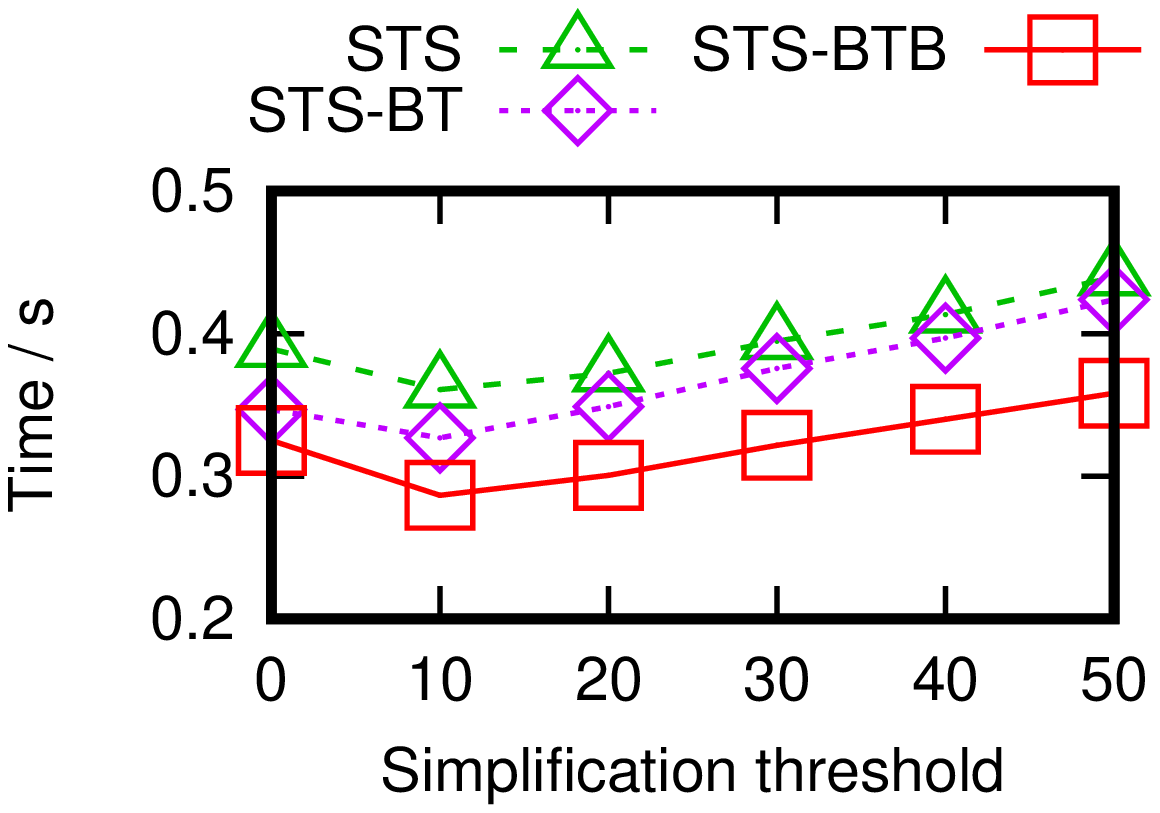}
    }
    \subfloat[{Number of node accesses}~\label{fig:exp4_2_io}]{
    	\includegraphics[width=0.23\textwidth]{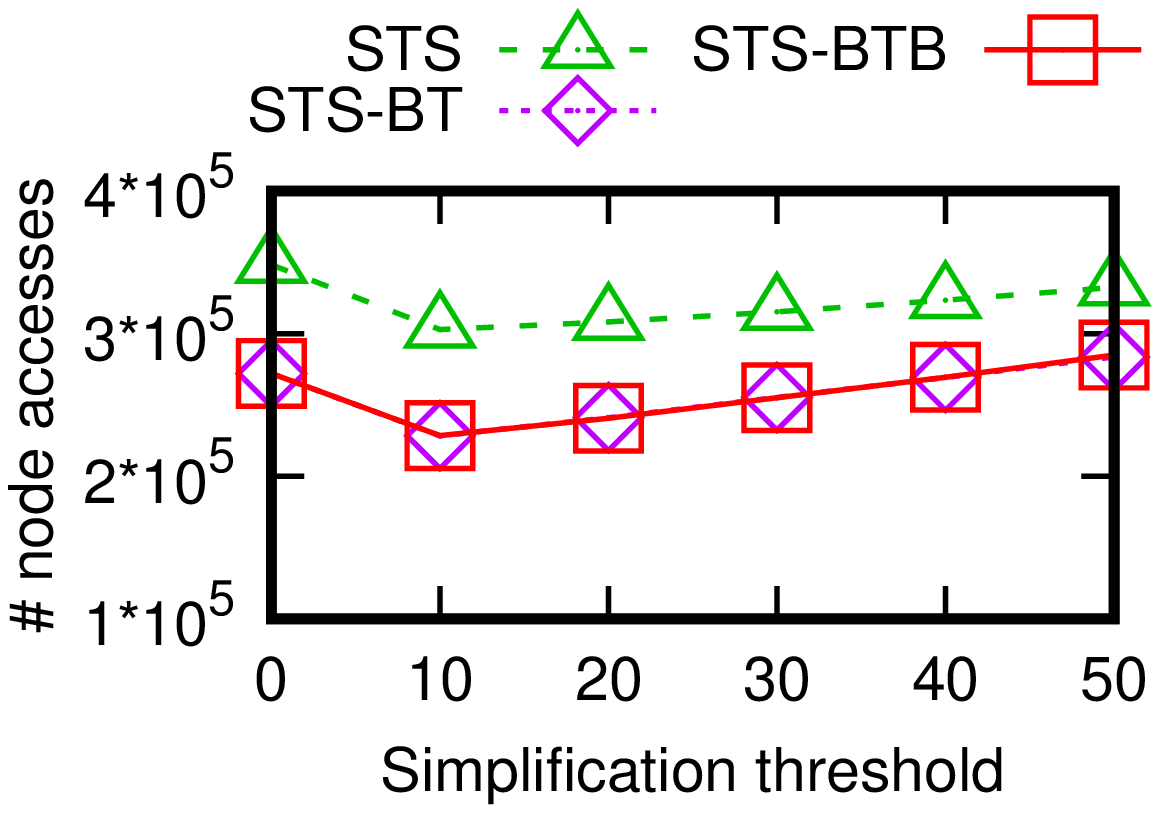}
    }\\
    \vspace{-3mm}
    \subfloat[{Number of segments sent}~\label{fig:exp8_2_segment}]{
    	\includegraphics[width=0.23\textwidth]{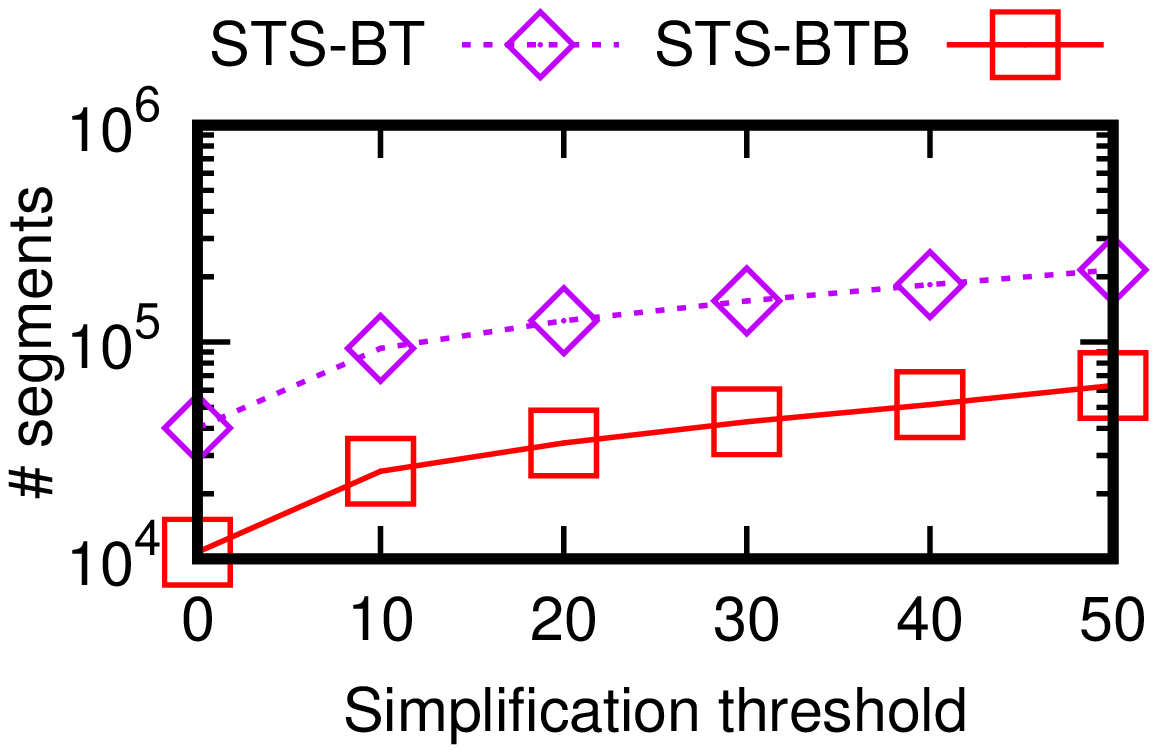}
    }
    \subfloat[Number of receiving clients~\label{fig:exp8_2_client}]{
    	\includegraphics[width=0.23\textwidth]{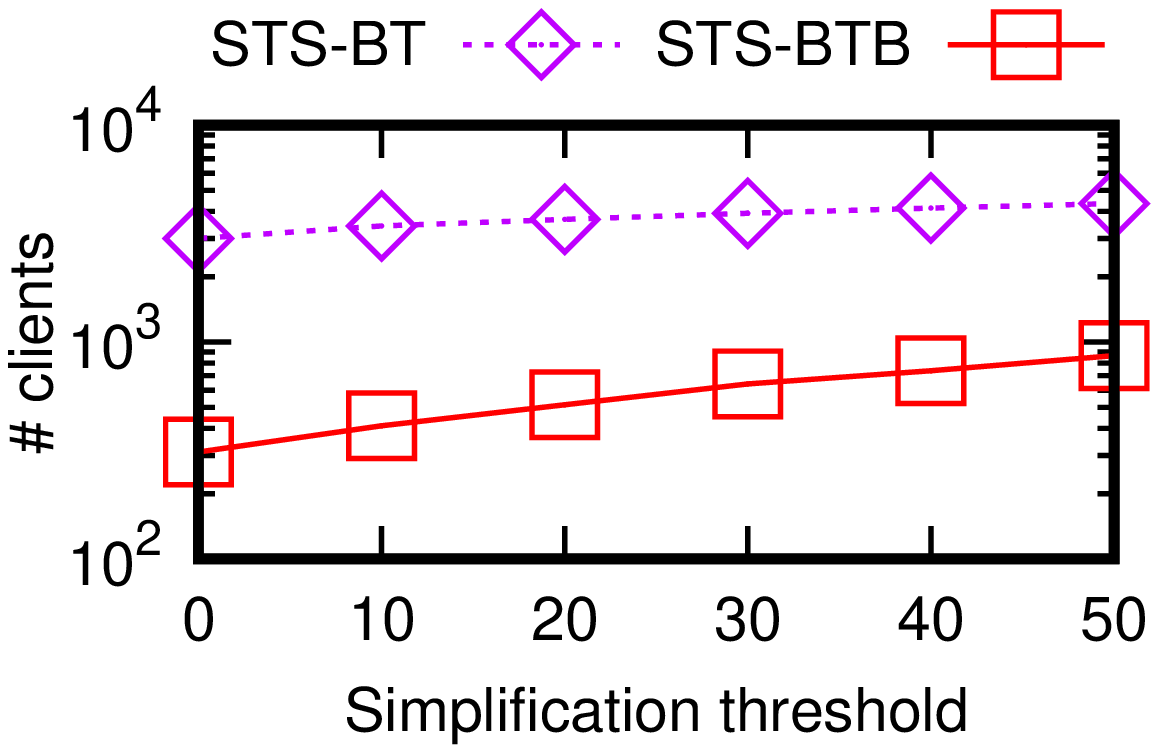}
    }
    \caption{Query costs vs. simplification threshold}\label{fig:join_simplification_threshold}
\end{figure}

\textbf{Varying the shifting distance $\theta_{ob}$. } 
Fig.~\ref{fig:join_shifting_distance} shows that the query costs of the STS algorithm variants increase with $\theta_{ob}$, as expected. 
A larger $\theta_{ob}$ may enlarge the node MBRs and hence leads to higher query costs. STS-BTB still takes the smallest time costs, number of node accesses, and  communication costs,  which have similar trends to those when varying $\theta_{sp}$, further confirming the robustness of the algorithm and the optimization techniques.

\begin{figure}[h]
    \vspace{-3mm}
    \centering
    \subfloat[{Response time}~\label{fig:exp4_3_time}]{
    	\includegraphics[width=0.23\textwidth]{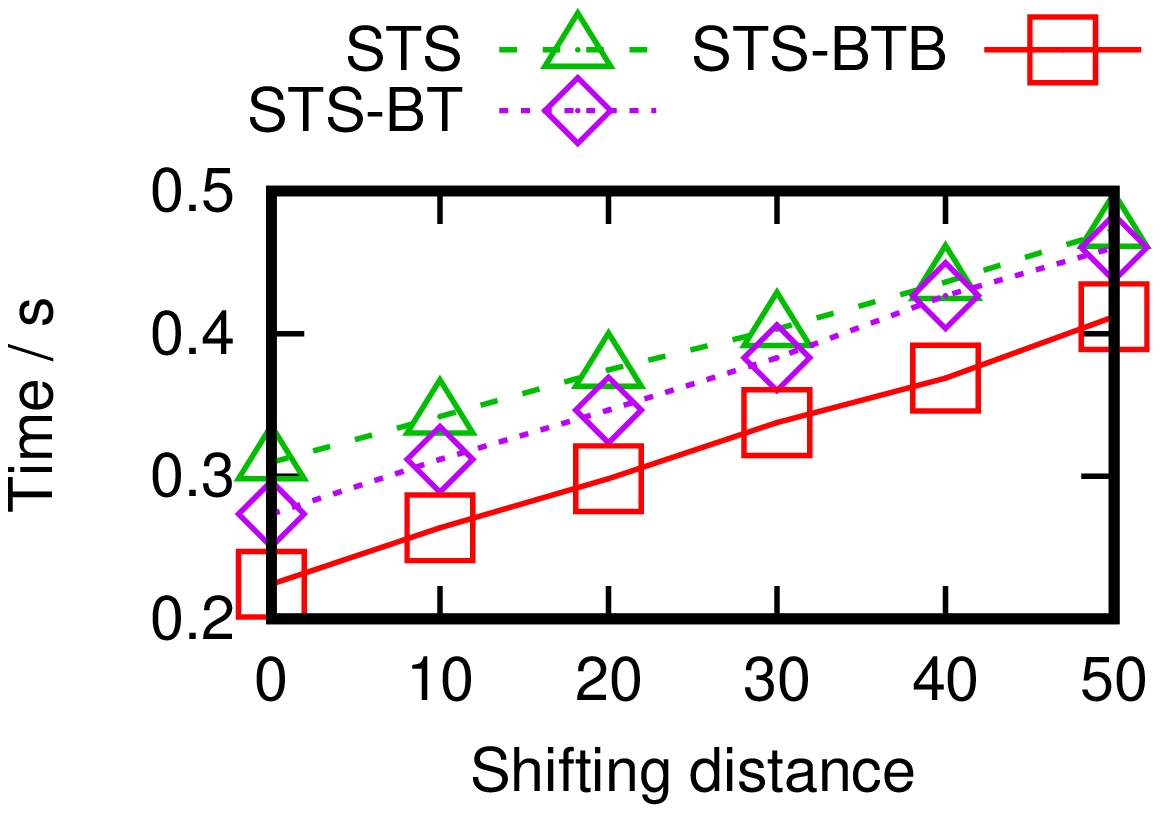}
    }
    \subfloat[{Number node accesses}~\label{fig:exp4_3_io}]{
    	\includegraphics[width=0.23\textwidth]{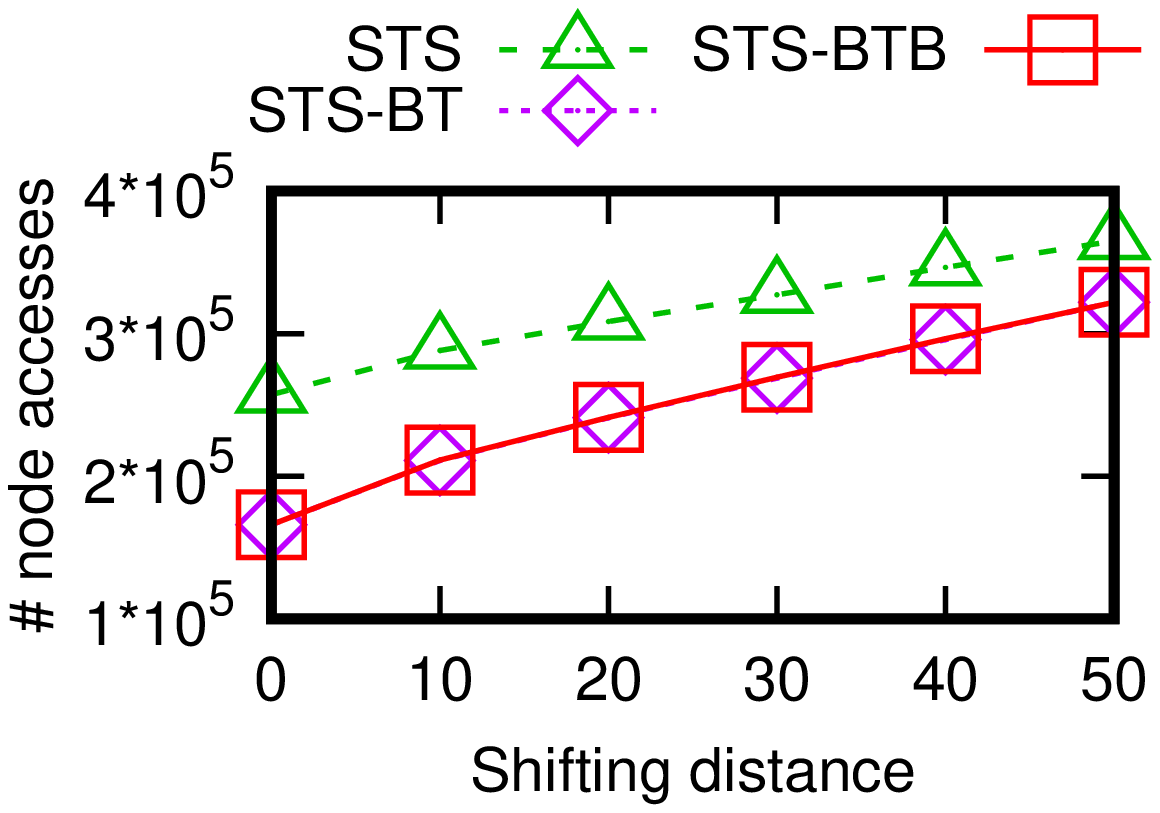}
    }\\
    \vspace{-2mm}
    \subfloat[{Number of segments sent}~\label{fig:exp8_3_segment}]{
    	\includegraphics[width=0.23\textwidth]{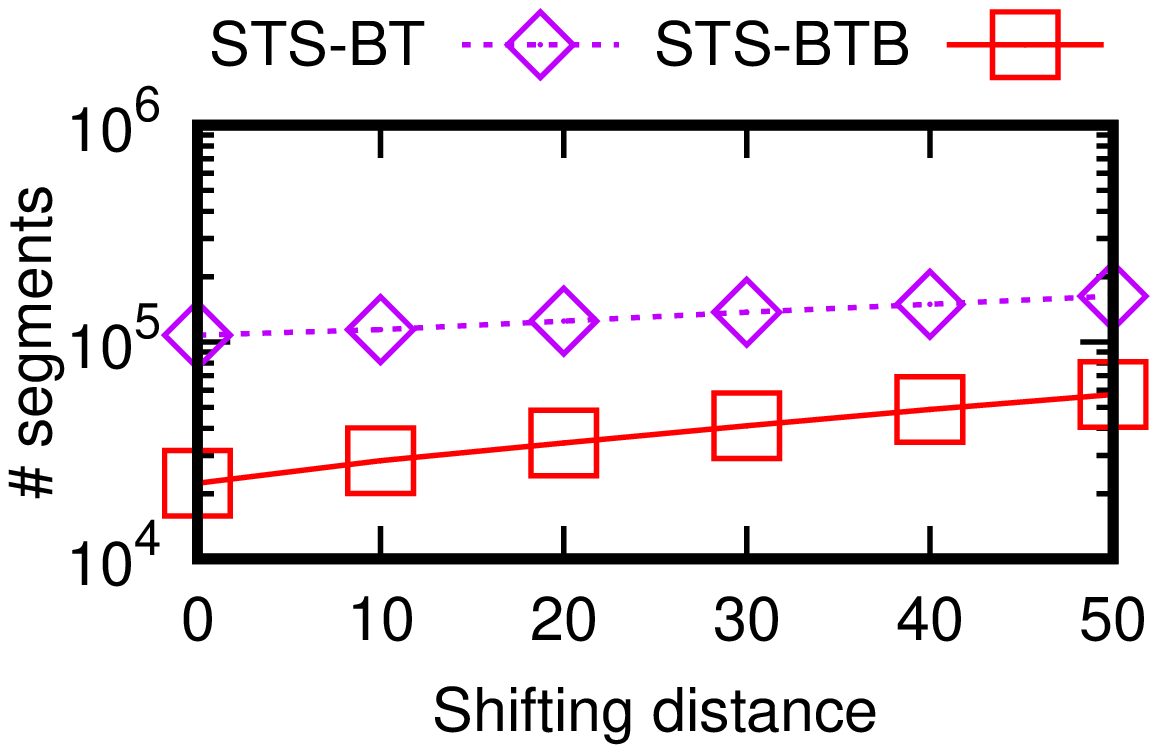}
    }
    \subfloat[Number of receiving  clients~\label{fig:exp8_3_client}]{
    	\includegraphics[width=0.23\textwidth]{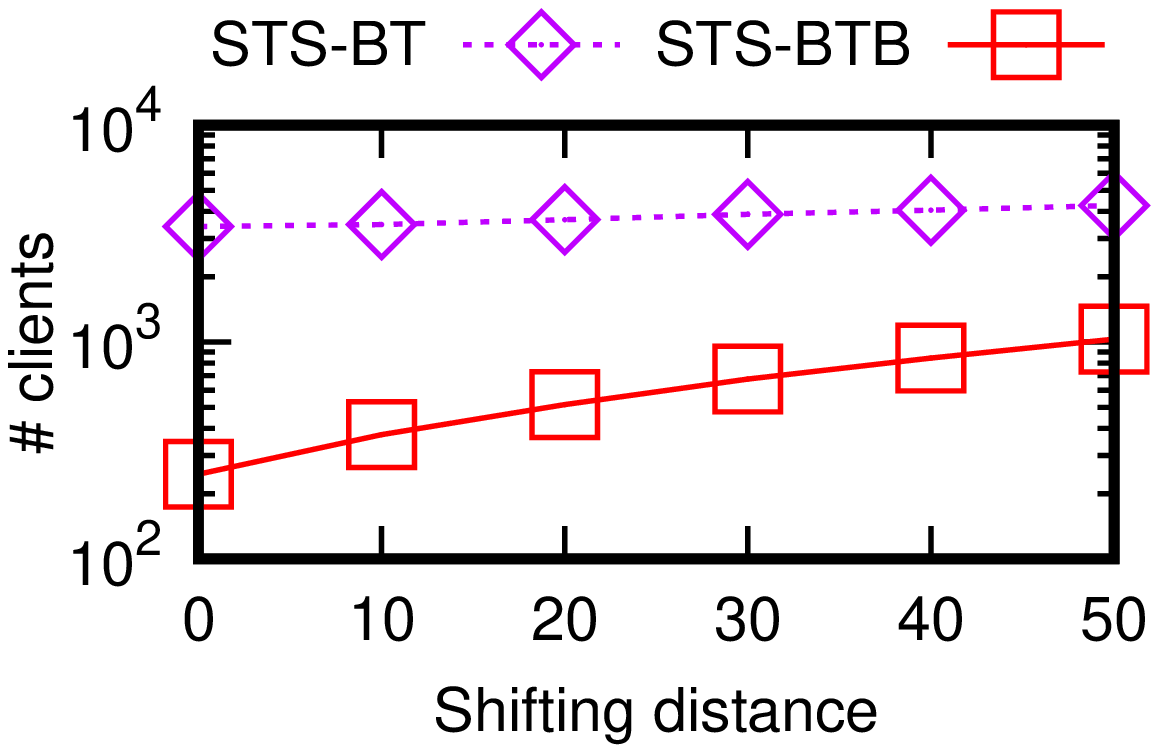}
    }
    \caption{Query costs vs. shifting distance}\label{fig:join_shifting_distance}
\end{figure}

When other STS-Index parameters are varied (the B-tree, and the bitmap time interval lengths), STS-BTB consistently outperforms STS-BT and STS. We omit the figures for conciseness.

\subsection{Index Construction Performance} 
We show the costs of index construction in Fig.~\ref{fig:index_dataset_size}. 
Since STS-BT and STS-BTB have the same process for index construction, we omit STS-BTB from the figure. 
STS-BT index is the fastest to build, which takes up to 15\% less time than ALSTJ-T.
STS-BT index also takes fewer node accesses to build than ALSTJ-T index does except when $|\mathcal{D}_p| \geqslant$ 400k
(recall that our index has smaller node capacities). 
STS takes more time and node accesses to build without backtracking, while ALSTJ suffers in its lack of pruning power in the time dimension for data insertion. 
3DR takes the largest costs to build, for that it is not optimized for trajectories. 

\begin{figure}[htp]
    \vspace{-6mm}
    \centering
    \subfloat[{Response time}~\label{fig:exp1_time}]{
    	\includegraphics[width=0.23\textwidth]{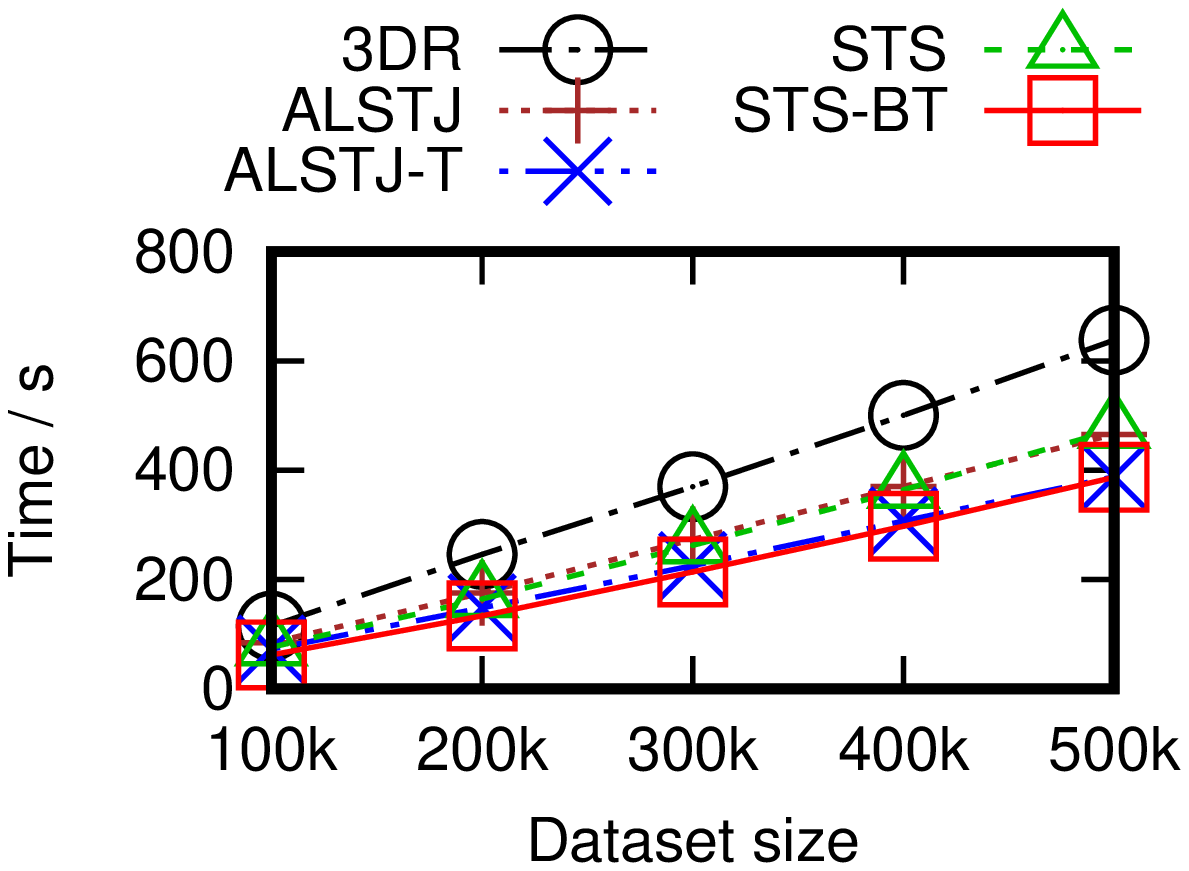}
    }
    \subfloat[{Number of node accesses}~\label{fig:exp1_io}]{
    	\includegraphics[width=0.23\textwidth]{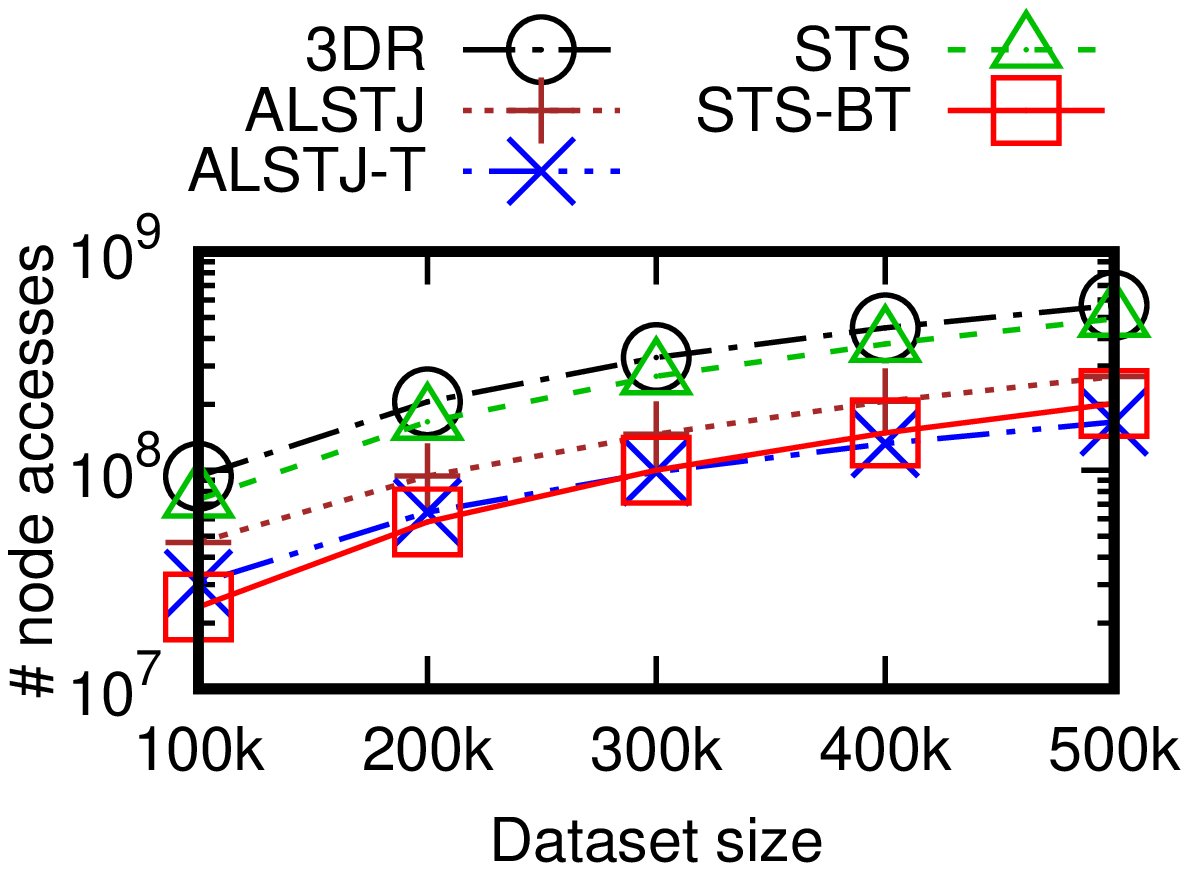}
    }\\
    \vspace{-2mm}
        \subfloat[{Index size}~\label{fig:exp_ram}]{
    	\includegraphics[width=0.23\textwidth]{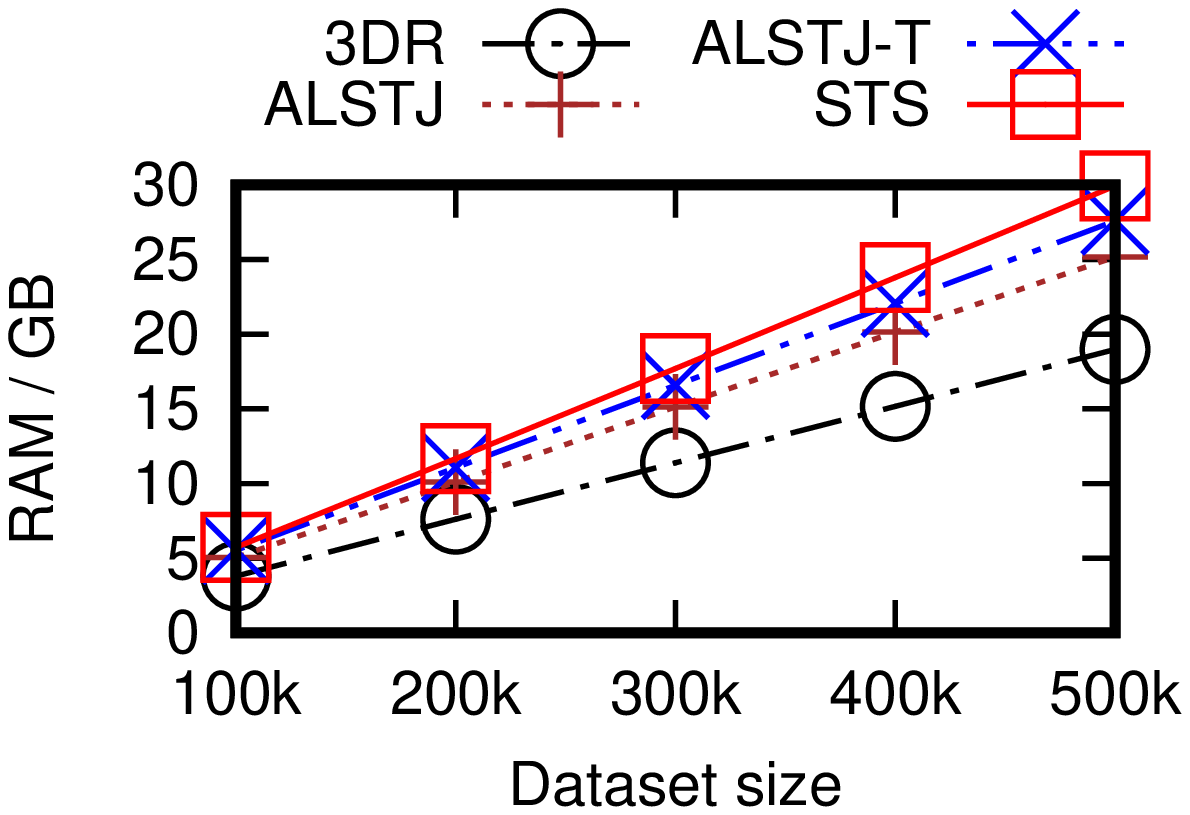}
    }
    \vspace{-1mm}
    \caption{{Index construction costs vs. dataset size}}\label{fig:index_dataset_size}
    \vspace{-1mm}
\end{figure}

We report the index sizes (RAM space) in Fig.~\ref{fig:exp_ram}.
STS and STS-BT indices have the same size, so we only report for STS.  
STS, ALSTJ, and ALSTJ-T indices are larger than those of 3DR, because they store simplified and original trajectories while 3DR only stores original trajectories.  
STS indices take slightly ($\leqslant$8.5\%) more space than ALSTJ-T does. This extra cost comes from the bitmap structure and more index nodes, which is justified by our much faster query performance. ALSTJ indices are smaller than ALSTJ-T's as they do not index the time dimension. 


\textbf{Effectiveness of the join predicate.} 
To show the differences between join queries, we run the original ALSTJ algorithm (with its own join predicate~\cite{join_subtraj_edbt2020}) and collect the resultant trajectory pairs. We find that this misses 36\% of the trajectory pairs returned by STS-Join, confirming the effectiveness of STS-Join in identifying trajectory pairs that may be missed by existing queries.

\vspace{-3mm}
\section{Conclusions} \label{sec:conclusion}
We propose a trajectory join query named STS-Join that returns pairs of trajectories within a given distance threshold 
for a time longer than a given time threshold. 
Unlike existing join queries that require the full trajectories to be similar, 
STS-Join can find trajectories that are similar for a time period while being quite different overall.  
We propose a client-server index named STS-Index and an efficient backtracking-based algorithm and a similarity time period based pruning strategy to support queries while supporting privacy protection. 
Our cost analysis and experiments on real data confirm the superiority of the proposed algorithm -- it outperforms adapted state-of-the-art join algorithms consistently and by up to three orders of magnitude in the query time. 
 

For future work, we plan to extend STS-Join to the road network, and allow temporal shifts in the similarity metric.
We also intend to further study the privacy protection on trajectory similarity queries from a theoretical perspective.

\vspace{-3mm}
\begin{acks}
This work was sponsored in part by the Australian Research Council under the Discovery Project Grant DP180103332 and by the NSF under Grants IIS-18-16889 and IIS-20-41415.
\end{acks}
\vspace{-3mm}

\bibliographystyle{ACM-Reference-Format}
\bibliography{ms}

\end{document}